\documentclass[a4paper]{article}

\usepackage[all]{xy}\usepackage[latin1]{inputenc}        
\usepackage[dvips]{graphics,graphicx}
\usepackage{amsfonts,amssymb,amsmath,color,mathrsfs, amstext}
\usepackage{amsbsy, amsopn, amscd, amsxtra, amsthm,authblk}
\usepackage{upref}
\usepackage[colorlinks,
            linkcolor=red,
            anchorcolor=red,
            citecolor=red
            ]{hyperref}

\usepackage{geometry}
\usepackage[displaymath]{lineno}
\usepackage{float}

\usepackage{yhmath}

\numberwithin{equation}{section}              
\newtheorem{theorem}{Theorem}[section]
\newtheorem{lemma}{Lemma}[section]
\newtheorem{proposition}{Proposition}[section]
\newtheorem*{proposition*}{Proposition}
\newtheorem{corollary}{Corollary}[section]
\newtheorem*{corollary*}{Corollary}
\newtheorem{definition}{Definition}[section]
\newtheorem*{definitions*}{Definitions}
\newtheorem*{conjecture*}{\bf Conjecture}

\newtheorem*{example*}{\bf Example}
\theoremstyle{remark}
\newtheorem{remark}{\bf Remark}[section]

\begin{document}
\date{}                                     
\title{Patched peakon weak solutions of the modified Camassa-Holm equation}

\author[1,2,3]{Yu Gao\thanks{yugao@hit.edu.cn}}
\author[2]{Lei Li\thanks{leili@math.duke.edu}}
\author[2,3]{Jian-Guo Liu\thanks{jliu@phy.duke.edu}}
\affil[1]{Department of Mathematics, Harbin Institute of Technology, Harbin, 150001, P.R. China.}
\affil[2]{Department of Mathematics, Duke University, Durham, NC 27708, USA.}
\affil[3]{ Department of Physics, Duke University, Durham, NC 27708, USA.}

\maketitle

\begin{abstract}
In this paper, we study traveling wave solutions and peakon weak solutions of the modified Camassa-Holm (mCH) equation with dispersive term $2ku_x$ for $k\in\mathbb{R}$.  We study traveling wave solutions through a Hamiltonian system obtained from the mCH equation by using a nonlinear transformation. The typical traveling wave solutions given by this Hamiltonian system are unbounded or multi-valued. We provide a method, called patching technic, to truncate these traveling wave solutions  and patch different segments to obtain patched bounded single-valued peakon weak solutions which satisfy jump conditions at peakons. Then, we study some special peakon weak solutions constructed by the fundamental solution of the Helmholtz operator $1-\partial_{xx}$, which can also be obtained by the patching technic.  At last, we study some length and total signed area preserving closed planar curve flows that can be described by the mCH equation when $k=1$, for which we give a Hamiltonian structure and use the patched  periodic peakon weak solutions to investigate loops with cusps.
\end{abstract}

\section{Introduction}

This paper is concerned with traveling wave solutions and peakon weak solutions to the following nonlinear partial differential equation:
\begin{align}\label{mCH}
m_t+2ku_x+[(u^2-u^2_x)m]_x=0,\quad m=u-u_{xx},\quad x\in\mathbb{R},~~t>0.
\end{align}
This equation is referred to as the modified Camassa-Holm (mCH) equation with cubic nonlinearity, which was introduced by several different authors \cite{Fokas,Fuchssteiner,Olever,Qiao2}. The parameter $k\in \mathbb{R}$ characterizes the magnitude of the linear dispersion. 
The mCH equation  is a water wave equation and a suitable approximation of the incompressible irrotational Euler system.
The functions $u$ and $m$ represent, respectively, the velocity of the fluid and its potential density \cite{Qiao2}.  The weak solutions can be non-zero or even unbounded as $|x|\to\infty$. If a weak solution is periodic with period $\ell$ where $0<\ell<+\infty$, then it can be regarded as a weak solution of the mCH equation on the torus $\mathbb{T}_\ell:=[-\ell/2,\ell/2)\cong\mathbb{R}/(\ell\mathbb{Z})$. We also denote $\mathbb{T}_\ell=\mathbb{R}$ when $\ell=+\infty$ for convenience.

When $\ell<+\infty$, the fundamental solution $G(x)$ for the Helmholtz operator $1-\partial_{xx}$  is given by
\begin{gather}\label{eq:fundamentalsolutionG1}
G(x)=\left\{
\begin{split}
&\frac{\cosh(x+\ell/2)}{2\sinh(\ell/2)},~\textrm{ for }~x\in [-\ell/2,0),\\
&\frac{\cosh(x-\ell/2)}{2\sinh(\ell/2)},~\textrm{ for }~x\in [0, \ell/2),
\end{split}
\right.
\end{gather}
 and when $\ell=+\infty$, we have
\begin{align}\label{eq:fundamental2}
G(x)=\frac{1}{2}e^{-|x|},~\textrm{ for }~x\in\mathbb{R}.
\end{align}
The velocity function $u$ can be written as a convolution of $m$ with  $G$
\begin{align*}
u(x,t)=G\ast m=\int_{\mathbb{T}_\ell}G(x-y)m(y,t)dy,~~x\in\mathbb{T}_\ell.
\end{align*}

For smooth solutions to the mCH equation \eqref{mCH}, there are two conserved quantities (called Hamiltonian functionals of the mCH equation) 
\begin{align}\label{Hamiltonians}
H_0=\int_{\mathbb{T}_\ell}mudx,\quad H_1=\frac{1}{4}\int_{\mathbb{T}_\ell}\Big(u^4+2u^2u_x^2-\frac{1}{3}u^4_x+4ku^2\Big)dx.
\end{align}
\eqref{mCH} can be written in the bi-Hamiltonian form \cite{Olever,GuiLiuOlverQu},
$$m_t=-((u^2-u_x^2)m)_x=J\frac{\delta H_0}{\delta m}=K\frac{\delta H_1}{\delta m},$$
where $$J=-\partial_xm\partial_x^{-1}m\partial_x-k\partial_x,\quad K=\partial_x^3-\partial_x$$
are compatible Hamiltonian operators. (If the  Hamiltonian operators $J$ and $K$ are compatible, then any constant coefficient linear combination $\alpha J +\beta K$ is also a Hamiltonian operator.) The Hamiltonian pair $J$,$K$ is non-degenerate in the sense that one of the associated Poisson structure is symplectic. According to the fundamental theorem of Magri \cite{Magri}, any bi-Hamiltonian system associated with a non-degenerate Hamiltonian pair induces a hierarchy of commuting Hamiltonians and flows, and, provided there are enough functionally independent Hamiltonians, is therefore completely integrable. The mCH equation \eqref{mCH} is a complete integrable system and it possesses a Lax pair and may be solved by means of the inverse scattering transform (IST) method  \cite{Qiao2}.

In this paper, one of our main purposes is to introduce a patching technic as described below to construct peakon weak solutions from traveling wave solutions to the mCH equation \eqref{mCH}. To this end, we first give some jump conditions (see Theorem \ref{thm:jumpcondition}) for  piecewise smooth weak solutions to \eqref{mCH}. Then, we study  traveling wave solutions in the form $u(x,t)=\phi(x-ct)=\phi(\xi)$ ($\xi:=x-ct$) with traveling speed $c\in\mathbb{R}$. Under the traveling wave assumption, Equation \eqref{mCH} can be reduced to a third order autonomous ODE (see \eqref{eq:3rdODE}). After a first integral, this third order ODE can be reduced to a second order ODE (see \eqref{eq:travellingwaveODE}). Then, after some transformation this second order ODE can be rewritten as a Hamiltonian system  with a Hamiltonian $H$ (see \eqref{eq:system2}). Then, we study the structure of level sets of $H$ in the phase plane (see Proposition \ref{pro:structureGam}). We show that traveling wave solutions to the mCH equation can be constructed from level sets of $H$ (see Proposition \ref{pro:leveltoSol}). Consider a level set $H=h$ for some constant $h$. There are two particlular cases for solutions corresponding to $H=h$: the case  $k\leq0$ and the case $k>0$. When $k\leq0$ and $h$ is big enough,  there is a trajectory in the phase plane given by $H=h$ extending to infinity (see Theorem \ref{thm:levelsets}) and this corresponds to a single-valued traveling wave solution to \eqref{mCH} that satisfies $\lim_{|\xi|\to +\infty}|\phi(\xi)|=+\infty$.
For the case $k>0$, the trajectories in the phase plane given by $H=h$ are some closed curves and when the traveling speed $c<0$ or $h$ big enougth for $c>0$, these closed curves yield multi-valued (see Definition \ref{def:multi}) periodic traveling wave solutions (see Theorem \ref{thm:levelsets}). 

When $k>0$, the patching technic is used to obtain a patched single-valued periodic traveling peakon weak solution from a multi-valued periodic traveling wave solution.  Briefly speaking, we truncate the level set $H=h$ by curves $\phi^2-\frac{1}{3}v^2=c$ in phase plane. The two intersection points on the same branch of $\phi^2-\frac{1}{3}v^2=c$ glue together and form a patched single-valued periodic traveling wave weak solution with peakon(s). This truncating method guarantees that the jump condition is satisfied at the peakons and hence a weak solution (see Picture \ref{fig:solutions}). In this paper, this method is called `the patching technic' which can also be used to obtain patched bounded traveling peakon weak solutions for the case $k\leq0$.  

By the patching technic, we can obtain some special peakon weak solutions of the form
\begin{align}\label{solitarywave}
u(x,t)=pG(x-ct),~~m(x,t)=p\delta(x-ct),~~x\in\mathbb{R},~t\geq0,
\end{align}
where $p\neq0$ is a constant describing the amplitude of a peakon and the traveling speed is given by $c=\frac{1}{6}p^2$ (see Remark \ref{remark:specialcase}). \eqref{solitarywave} is a peakon weak solution to the dispersionless mCH equation \eqref{mCH} (with $k=0$) in $\mathbb{R}$ and it was also obtained in \cite{GuiLiuOlverQu}. 
On the periodic domain $\mathbb{T}_\ell$ ($\ell<+\infty$),  the mCH equation also has peakon weak solutions with similar form as \eqref{solitarywave} but the traveling speed is given by (see Proposition \ref{pro:periodicpeakon})
\[
c=\frac{1}{4}p^2\Big[\coth^2\left(\frac{\ell}{2}\right)-\frac{1}{3}\Big].
\]
As $\ell\rightarrow\infty$, this speed is consistent with the speed $\frac{1}{6}p^2$ in $\mathbb{R}$.   Notice that a solution given by \eqref{solitarywave} satisfies $m(x,t)=0$ in the smooth region of $u$ and the superposition of such solutions gives a new solution in the smooth region. Hence, the superposition of $N$ such solutions with different traveling speeds can form an $N$-peakon weak solution if the jump conditions are satisfied along the trajectories of peakons (see Proposition \ref{pro:NpeakonmCH}). 
 
When $k<0$, there are peakon weak solutions similar to \eqref{solitarywave} of the form (see Proposition \ref{pro:peakonknegative})  
\begin{align}\label{eq:travelknegative}
u(x,t)=pG(x-ct)-\sqrt{-k}, ~m(x,t)=p\delta(x-ct)-\sqrt{-k}.
\end{align}
When $\ell<\infty$, the travelling speed is given by
\[
 c=\Big[\frac{p}{2}\coth\left(\frac{\ell}{2}\right)-\sqrt{-k}\Big]^2-\frac{p^2}{12},
\]
and when $\ell=+\infty$, it is given by
\[ c=\Big(\frac{p}{2}-\sqrt{-k}\Big)^2-\frac{p^2}{12}.
\]
This kind of peakon weak solutions can also be obtained by the patching technic (see Remark \ref{rmk:specialknegative}).  Notice that in the whole space $\mathbb{R}$, this kind of solutions will not vanish as $x\rightarrow \infty.$   In \cite{lwj04,Qian}, some similar peakon weak solutions are obtained for the Camassa Holm (CH) equations  with dispersive term $2ku_x$ (for any $k\in \mathbb{R}$):
\begin{align}\label{eq:generalCH}
u_t-u_{xxt}+3uu_x+2ku_x=2u_xu_{xx}+uu_{xxx},~~ x\in\mathbb{R},~t>0.
\end{align}

For the case  $k>0$, the mCH equation \eqref{mCH} does not have peakon weak solutions of the form similar as \eqref{solitarywave} or \eqref{eq:travelknegative} (see Proposition \ref{pro:nopeakonkpositive}). However, the mCH equation \eqref{mCH} with $k>0$ has peakon weak solutions in other forms, which can be obtained by the patching technic (see Figure \ref{fig:solutions}).

Specially, when $k=1$, the mCH equation \eqref{mCH} can be used to describe planar curve flows which preserve arc-length \cite{GuiLiuOlverQu}. In this paper, we show  that  in periodic case, the mCH equation ($k=1$) describes some dynamics of closed plane curves  whose length and total signed area (defined by \eqref{eq:totalsignedarea}) are both preserved. This property is also shared with the modified Korteweg-de Vries equaiton. Moreover, we give a Hamiltonian structure for the curve dynamics (see Proposition \ref{pro:variation1}) and use the patched peakon weak solutions to investigate loops with cusps (see Figure \ref{fig:curve}).

For more results about local well-posedness and blow-up behaviors of the strong solutions to  \eqref{mCH}, one can refer to  \cite{Chenrongming,FuGui,GuiLiuOlverQu,Himonas,LiuOlver}.
In \cite{Qingtianzhang}, Zhang used a method of dissipative approximation to prove the existence and uniqueness of global entropy weak solutions $u$ in $W^{2,1}(\mathbb{R})$ for the dispersionless mCH equation \eqref{mCH} ($k=0$). In \cite{GaoLiu}, a sticky particle model is provided and a solution to this model yields a sticky $N$-peakon weak solution  which is a superposition of $N$ peakons in the form \eqref{solitarywave}. By some space-time BV estimates, the mean field limit of this model gives a global weak solution $u$ to the dispersionless mCH equation in $\mathbb{R}$ and $u,~u_x$ are space-time BV functions. When $k>0$,  Matsuno \cite{Matsuno} studied multi-soliton solutions to \eqref{mCH} in $\mathbb{R}$ by using the Hirota bilinear transformation method \cite{Hirota}. Both smooth and singular multi-soliton solutions are studied. The underlying structure of the associated tau-functions constituting the $N$-soliton solution is the same as that of the $N$-soliton solution of a model equation for shallow-water waves introducted by Ablowitz et al \cite{Ablowitz} (see \cite[Section 2]{Matsuno}).

The rest of this article is organized as follows. In Section \ref{sec:Jumpcondition}, we give the definition of weak solutions  to the mCH equation \eqref{mCH} and obtain some jump conditions for piecewise smooth weak solutions.   In Section \ref{sec:travelwave},   we deduce an ODE system to study traveling wave solutions to \eqref{mCH}. Traveling wave solutions are given by the level set of the first integral of the ODE system.  Then, we construct peakon weak solutions from  traveling wave solutions  by the patching technic.  In Section \ref{sec:peakonsolutions}, we show a special class of peakon weak solutions to \eqref{mCH} when $k\leq0$. These peakon weak solutions can also be obtained by the patching technic given in Section \ref{sec:travelwave}.
In Section \ref{sec:curve},  we give some criteria for preserving length and total signed area for closed curve flows. Moreover, we show that some length and total signed area preserving curve flows in $\mathbb{R}^2$ can be described by the mCH equation with $k=1$. These curve flows have a Hamiltonian structure. Lastly, we use the patched traveling peakon solutions in Section \ref{sec:travelwave} to investigate loops with cusps.

\section{Jump conditions for piecewise smooth weak solutions}\label{sec:Jumpcondition}

The physical solutions of \eqref{mCH} are usually piecewise smooth and contain some non-smooth points. Due to the non-smooth points, these solutions must be understood in the weak sense.  In this section, we first introduce the definition of weak solutions and obtain jump conditions for piecewise smooth weak solutions to the mCH equation \eqref{mCH}. Our patching technic in Section \ref{sec:patchedweaktravel} is based on these jump conditions.

\begin{definition}\label{weaksolution}
For $\ell<+\infty$ and $u_0\in W^{1,\infty}(\mathbb{T}_\ell)$, a function
\begin{align*}
u\in L^\infty(0,T; W^{1,\infty}(\mathbb{T}_\ell))\cap  C([0,T);H^1(\mathbb{T}_\ell)) ,
\end{align*}
is  a weak solution of the mCH equation \eqref{mCH} subject to initial data $u(x,0)=u_0(x)$ if
\begin{gather}\label{eq:weaksolutiondef}
\langle u, \varphi\rangle:=\mathcal{L}(u,\varphi)+\int_{\mathbb{T}_\ell} u_0 \varphi(x, 0)dx-\int_{\mathbb{T}_\ell} u_0\varphi_{xx}(x,0)dx=0
\end{gather}
where
\begin{align*}
\mathcal{L}(u,\varphi):=&\int_0^T\int_{\mathbb{T}_\ell} u(\varphi_t-\varphi_{xxt}) dxdt
+\int_0^T\int_{\mathbb{T}_\ell} (2ku+u^3+uu_x^2)\varphi_xdxdt\\
&-\frac{1}{3}\int_0^T\int_{\mathbb{T}_\ell}u^3\varphi_{xxx}dxdt-\frac{1}{3}\int_0^T\int_{\mathbb{T}_\ell} u_x^3\varphi_{xx}dxdt.
\end{align*}
holds for all $\varphi\in C_c^{\infty}([0,T)\times\mathbb{T}_\ell)$ .  If $T=+\infty$, we call $u$ as a global weak solution of the mCH equation. 

When $\ell=+\infty$ ($\mathbb{T}_\ell=\mathbb{R}$) and $u_0\in W_{loc}^{1,\infty}(\mathbb{R})$, we say a function
\begin{align*}
u\in L^\infty(0,T; W_{loc}^{1,\infty}(\mathbb{R}))\cap  C([0,T);H_{loc}^1(\mathbb{R})),
\end{align*}
is a weak solution of the mCH equation \eqref{mCH} subject to initial data $u(x,0)=u_0(x)$ if \eqref{eq:weaksolutiondef} holds for all $\varphi\in C_c^{\infty}([0,T)\times\mathbb{R})$. 
\end{definition}

It is straightforward to check that the definition is consistent. In other words, a weak solution on $\mathbb{T}_{\ell}$ can be identified as a weak solution on $\mathbb{T}_{n\ell}$ for $n=1,2 \ldots$ or $\mathbb{R}$ if we extend the function periodically. Reversely, a periodic weak solution on $\mathbb{R}$ with period $\ell$, (or a periodic weak solution on $\mathbb{T}_{\ell_1}$ with smaller period $\ell<\ell_1$) can be regarded as a weak solution on $\mathbb{T}_{\ell}$. This consistency allows us to identify the periodic weak solutions to \eqref{mCH} with weak solutions on $\mathbb{T}_{\ell}$.

Now, we state and prove the following results about jump conditions in this section.
\begin{theorem}\label{thm:jumpcondition}
Assume $0<\ell<+\infty$.
Suppose the region $\mathbb{T}_\ell\times [0, T]$ is divided by $N$ (no intersection) curves $\{(x_i(t),t): 0\leq t\leq T\}_{i=1}^N$ into $N$ non-overlap open area $\{U_i\}_{i=1}^N$, where $N$ is a positive integer.   The boundary between $U_i$ and $U_{i+1}$ is $x_i(t)$ and $x_N(t)$ is the boundary between $U_N$ and $U_{N+1}=U_1$.  Assume that $u_i\in C^3({U}_i)\cap C^1(\bar{U}_i)$ solves the mCH equation with $u_i(x,0)=u_0(x)$ in $U_i$ and $m_i=u_i-\partial_{xx}u_i$ for $1\leq i\leq N$.  Denote
$$v_\ell^i(t):=\partial_xu_i(x_i(t), t),~~ v_r^i(t):=\partial_xu_{i+1}(x_i(t), t)$$
and assume the following two limits exist
\begin{gather*}
\left\{
\begin{split}
&\lim_{x\to x_i(t)^-} m_i\Big[u^2_i-(\partial_xu_i)^2-x'_i(t)\Big]=:A^i_{\ell}(t),\\
&\lim_{x\to x_i(t)^+} m_{i+1}\Big[u^2_{i+1}-(\partial_xu_{i+1})^2-x'_i(t)\Big]=:A_r^{i}(t).
\end{split}
\right.
\end{gather*}
Then, the function $u$ satisfying $u|_{U_i}=u_i$ is a weak solution to the mCH equation subject to $u(x,0)=u_0(x)$ if the following conditions are satisfied:
\begin{enumerate}
\item $u\in C(\mathbb{T}_\ell\times[0,T))$.
\item For each $i$, it holds that 
\begin{gather}\label{eq:jumpcd1}
\left\{
\begin{split}
&\frac{d}{dt}x_i(t)=u(x_i(t), t)^2-\frac{1}{3}[v^i_\ell(t)]^2-\frac{1}{3}v^i_\ell (t)v^i_r(t)-\frac{1}{3}[v^i_r(t)]^2,\\
& \frac{d}{dt}[v_\ell^i(t)-v_r^i(t)]=A_\ell^i(t)-A_r^i(t)
\end{split}
\right.
\end{gather}
or
\begin{gather}\label{eq:continu}
v^i_{\ell}=v^i_r,~ A_{\ell}^i=A_r^i.
\end{gather}
\end{enumerate}
When $\ell=+\infty$ ($\mathbb{T}_\ell=\mathbb{R}$), we have the same jump conditions.
\end{theorem}

\begin{proof}
By the fact $u_i\in C^1(\bar{U}_i)$ and Condition 1, it is clear that
$$u\in  L^\infty(0,T; W^{1,\infty}(\mathbb{T}_\ell))\cap C([0,T);H_{loc}^1(\mathbb{T}_\ell)).$$

For the Condition 2, we first assume \eqref{eq:jumpcd1}. Take $\varphi(x,t)\in C_c^\infty([0,T)\times\mathbb{T}_\ell)$. For initial data,  we have
\begin{multline}\label{initial}
\int_{\mathbb{T}_\ell} u_0\varphi(x,0)dx-\int_{\mathbb{T}_\ell} u_0\varphi_{xx}(x,0)dx
 =\sum_{i=1}^N\int_{x_{i-1}(0)}^{x_i(0)} m_i(x,0)\varphi(x,0)dx\\
 +\sum_{i=1}^N\varphi(x_i(0),0)(v_\ell^i(0)-v_r^i(0)),
\end{multline}
where $x_0=x_N.$
From Definition \ref{weaksolution}, we have
\begin{multline}\label{eq:L0}
\mathcal{L}(u,\varphi)=\int_0^T\int_{\mathbb{T}_\ell} u(\varphi_t-\varphi_{xxt}) dxdt+\int_0^T\int_{\mathbb{T}_\ell} (2ku+u^3+uu_x^2)\varphi_xdxdt\\
-\frac{1}{3}\int_0^T\int_{\mathbb{T}_\ell}u^3\varphi_{xxx}dxdt-\frac{1}{3}\int_0^T\int_{\mathbb{T}_\ell} u_x^3\varphi_{xx}dxdt.
\end{multline}
With some calculations, we obtain
\begin{multline}\label{eq:L1}
\int_0^T\int_{\mathbb{T}_\ell} u(\varphi_t-\varphi_{xxt}) dxdt=\int_0^T\int_{\mathbb{T}_\ell} u\varphi_t dxdt+\sum_{i=1}^N\int_0^T\int_{x_{i-1}(t)}^{x_i(t)} u_x\varphi_{xt}dxdt \\
=\int_0^T\sum_{i=1}^N(v^i_\ell(t)-v_r^i(t))\varphi_t(x_i(t),t)dt+\sum_{i=1}^N\iint_{U_i}m_i(x,t)\varphi_t dxdt
\end{multline}
and
\begin{multline}\label{eq:L2}
-\frac{1}{3}\int_0^T\int_{\mathbb{T}_\ell}u^3\varphi_{xxx}dxdt-\frac{1}{3}\int_0^T\int_{\mathbb{T}_\ell} u_x^3\varphi_{xx}dxdt \\
=\sum_{i=1}^N\int_0^T (v_\ell^i-v^i_r)\varphi_x(x_i(t),t)\frac{d}{dt}x_i(t) dt-\sum_{i=1}^N\iint_{U_i}(2uu_x^2+u^2u_{xx}-u_x^2u_{xx})\varphi_xdxdt,
\end{multline}
where we used the jump condition \eqref{eq:jumpcd1} for the last step. On each $U_i$, $u_{xx}$ agrees with $(u_i)_{xx}$, which is a continuous function.
Combining \eqref{eq:L0}, \eqref{eq:L1} and \eqref{eq:L2} gives
\begin{align*}
\mathcal{L}(u,\varphi)=&\sum_{i=1}^N\iint_{U_i}m_i(x,t)\varphi_t dxdt+\sum_{i=1}^N\iint_{U_i}(2ku+m_i(u^2-u_x^2))\varphi_xdxdt\\
&+\sum_{i=1}^N\int_0^T \Big[(v^i_\ell(t)-v_r^i(t))\varphi_t(x_i(t),t)+(v_\ell^i(t)-v^i_r(t))\varphi_x(x_i(t),t)\frac{d}{dt}x_i(t)\Big] dt\\
=&:I_1+I_2+I_3.
\end{align*}

For $I_1+I_2$,  we have 
\begin{align*}
I_1+I_2&=\sum_{i=1}^N \iint_{U_i} \Big(m_i\varphi_t+(2k u+m_i(u^2-u_x^2))\varphi_x\Big) dxdt\\
&=\sum_{i=1}^N \iint_{U_i} \Big((m_i\varphi)_t+((2k u+m_i(u^2-u_x^2))\varphi)_x\Big) dxdt\\
&=\sum_{i=1}^N\oint_{\partial U_i} -m_i\varphi dx+(2ku_i+m_i(u_i^2-u_{ix}^2))\varphi dt.
\end{align*}
The boundary integral on $t=0$ together with the initial value terms yields 
\[
-\sum_{i=1}^N\int_{x_{i-1}(0)}^{x_i(0)} m_i(x,0)\varphi(x,0)dx.
\]
 The integral on $t=T$ is zero since $\varphi$ vanishes there. By the continuity of $u$ and the jump condition, we then find 
\begin{align*}
&I_1+I_2=-\sum_{i=1}^N\int_{x_{i-1}(0)}^{x_i(0)} m_i(x,0)\varphi(x,0)dx +\sum_{i=1}^N \int_0^T\Big[m_i x_{i-1}'(t)-m_i(u_i^2-u_{ix}^2)\Big]\varphi\Big|_{x=x_{i-1}} dt\\
&-\sum_{i=1}^N\int_0^T\Big[m_i x_{i}'(t)-m_i(u_i^2-u_{ix}^2)\Big]\varphi\Big|_{x=x_{i}(t)}  dt\\
=&-\sum_{i=1}^N\int_{x_{i-1}(0)}^{x_i(0)} m_i(x,0)\varphi(x,0)dx+\sum_{i=1}^N\int_0^T\varphi(x_i(t),t)(A_\ell^i-A_r^i)dt.
\end{align*}
By the fact
\[I_3=\sum_{i=1}^N\int_0^T (v^i_\ell(t)-v_r^i(t))\frac{d}{dt}\varphi(x_i(t),t)dt,\]
\eqref{initial} and Condition 2, we have:
\begin{align*}
\mathcal{L}(u,\phi)=&-\sum_{i=1}^N\int_{x_{i-1}(0)}^{x_i(0)} m_i(x,0)\varphi(x,0)dx+\sum_{i=1}^N\int_0^T\frac{d}{dt}\Big[\varphi(x_i(t),t)(v_\ell^i(t)-v_r^i(t))\Big]dt\\
=&-\sum_{i=1}^N\int_{x_{i-1}(0)}^{x_i(0)} m_i(x,0)\varphi(x,0)dx-\sum_{i=1}^N\varphi(x_i(0),0)(v_\ell^i(0)-v_r^i(0))\\
=&-\int_{\mathbb{T}_\ell} u_0\varphi(x,0)dx+\int_{\mathbb{T}_\ell} u_0\varphi_{xx}(x,0)dx.
\end{align*}
Therefore, by Definition \ref{weaksolution}, $u$ is a weak solution.

For the second condition, if we instead assume \eqref{eq:continu}, 
we find 
\begin{multline*}
\mathcal{L}(u,\varphi)=\sum_{i=1}^N\iint_{U_i}m_i(x,t)\varphi_t dxdt+\sum_{i=1}^N\iint_{U_i}(2ku+m_i(u^2-u_x^2))\varphi_xdxdt\\
=-\sum_{i=1}^N\int_{x_{i-1}(0)}^{x_i(0)} m_i(x,0)\varphi(x,0)dx
=-\int_{\mathbb{T}_\ell} u_0\varphi(x,0)dx+\int_{\mathbb{T}_\ell} u_0\varphi_{xx}(x,0)dx.
\end{multline*}
Hence it is also a weak solution.

For the case $\mathbb{T}_\ell=\mathbb{R}$ ($\ell=+\infty$), the proof is similar and we omit it.
\end{proof}

\section{Traveling wave solutions of the mCH equation}\label{sec:travelwave}
Consider traveling wave solutions of the form $u(x, t)=\phi(\xi)$, where $\xi:=x-ct$ and $c\in\mathbb{R}$ is a traveling speed. Hence, $m=u-u_{xx}=\phi-\phi''$. Then, the mCH equation is reduced to the following ODE:
 \begin{align}\label{eq:3rdODE}
c(\phi'''-\phi')+2k\phi'+\left((\phi-\phi'')(\phi^2-(\phi')^2)\right)'=0 .
\end{align}
Integrating once yields the following equation
\begin{gather}\label{eq:travellingwaveODE}
(c-\phi^2+(\phi')^2)\phi''-c\phi+2k\phi+\phi(\phi^2-(\phi')^2)=g,
\end{gather}
where $g$ is an integrating constant.

If $\phi$ is a solution for $g$, then $-\phi$ is a solution for $-g$. Therefore, the structure of solutions with positive $g$ will be the same as that of the solutions for negative $g$. Hence, we assume from here on that
\begin{gather}
g\ge 0.
\end{gather} 
Introducing $v=\phi'$, from \eqref{eq:travellingwaveODE} we can deduce the following first order system 
\begin{gather}\label{eq:system}
\left\{
\begin{split}
&\frac{d\phi}{d\xi}=v,\\
&\frac{dv}{d\xi}=\frac{g-2k\phi}{c-\phi^2+v^2}+\phi.
\end{split}
\right.
\end{gather}
Hence, we have
\begin{align*}
\frac{d\phi}{dv}=\frac{v(c-\phi^2+v^2)}{g-2k\phi+\phi(c-\phi^2+v^2)},
\end{align*}
which implies
\begin{align*}
(cv+v^3)dv-(g-2k\phi+c\phi-\phi^3)d\phi-\frac{1}{2}d(\phi^2v^2)=0.
\end{align*}
Hence, the first integral of System \eqref{eq:system} is given by
\begin{align}\label{eq:Hamiltonian}
H=\frac{(\phi^2-v^2)^2}{4}-\frac{1}{2}c(\phi^2-v^2)+k\phi^2-g\phi.
\end{align}
Changing of variable  $d\xi=(c-\phi^2+v^2)d\tau$, we have the following Hamiltonian system
\begin{gather}\label{eq:system2}
\left\{
\begin{split}
&\frac{d\phi}{d\tau}=\frac{\partial H}{\partial v}=v(c-\phi^2+v^2),\\
&\frac{dv}{d\tau}=-\frac{\partial H}{\partial \phi}=g-2k\phi-\phi(\phi^2-v^2-c).
\end{split}
\right.
\end{gather}

Because System \eqref{eq:system} and \eqref{eq:system2} have the same first integral $H$,  the two systems have the same topological phase portraits except the hyperbola $\phi^2-v^2-c=0$.  The level sets of $H$ are solution trajectories of System \eqref{eq:system2} and hence give traveling wave solutions of the mCH equation (except the hyperbola $\phi^2-v^2=c$).   In the following, we first discuss the critical points and level sets of the first integral $H$ . Then, we deduce traveling wave solutions from level sets of $H$. At last, we provide a patching technic to construct single-valued traveling peakon weak solutions from traveling wave solutions given by level sets of $H$.

\subsection{Critical points and level sets}

The critical points of System \eqref{eq:system2} in the phase plane are important for determining the structures of the level sets of $H$. The following observation enables us to classify the critical points of   System \eqref{eq:system2} and the proof is provided in Appendix \ref{app:critical}.
\begin{lemma}\label{lmm:critical}
Suppose the Hessian of $H$ is non-degenerate at a critical point $(\phi_*, v_*)$ of System \eqref{eq:system2}. Then, $(\phi_*, v_*)$ is a local extremum of $H$ if and only if it is a center points of System \eqref{eq:system2} in the phase plane while it is a saddle points of $H$ if and only if it is saddle point in the phase plane of System \eqref{eq:system2}.
\end{lemma}

As mentioned in the above,  Systems \eqref{eq:system} and \eqref{eq:system2}  have the same topological phase portraits except the hyperbola curves $\phi^2-v^2=c$. The critical points on  curves $\phi^2-v^2=c$ for system \eqref{eq:system2} can be summarized as follows (we only consider $k\neq 0$ here and the $k=0$ case is simple) and the proof is again provided in Appendix \ref{app:critical}:
\begin{proposition}\label{pro:criticalonhyperbola}
When $g^2-4k^2c<0$, there is no critical point on $\phi^2-v^2=c$. 
When  $g^2-4k^2c> 0$, there are two critical points for System \eqref{eq:system2} on $\phi^2-v^2=c$. If $k>0$, they are center points. If $k<0$, they are saddle points.
\end{proposition}

The critical points on $\phi^2-v^2=c$ do not give stationary solutions to System \eqref{eq:system}, but they determine the local structures of the level sets near them.  Other critical points of \eqref{eq:system2} are also critical points of \eqref{eq:system}, and they correspond to stationary solutions of System \eqref{eq:system} or constant traveling wave solutions of the mCH equation.  The critical points of System \eqref{eq:system2}  that are not on $\phi^2-v^2=c$ are of the form $(\phi_*, 0)$ ($\phi_*^2\neq c$),  where
\begin{gather}
-\phi_*^3-(2k-c)\phi_*+g=0.
\end{gather}
At this time, we have
$$\frac{\partial^2 H}{\partial v\partial\phi}(\phi_*,0)=0,~~\frac{\partial^2 H}{\partial\phi^2}(\phi_*,0)=3\phi_*^2+2k-c,~~\frac{\partial^2 H}{\partial v^2}(\phi_*,0)=c-\phi_*^2.$$
Hence
$$\Big(\frac{\partial^2 H}{\partial v\partial\phi}\Big)^2-\frac{\partial^2 H}{\partial\phi^2}\frac{\partial^2 H}{\partial v^2}=(3\phi_*^2+2k-c)(\phi_*^2-c)=-(c+4k)\phi_*^2+3g\phi_*-c(2k-c).$$
By choosing different $c,k$ and $g$, $(\phi_*,0)$ can be either a center point or a saddle point. 

We do not plan to give a thorough discussion for these critical points. The discussion is straightforward.  For example, in the case $g=0$ and $c-2k>0$,  there are three critical points with $v=0$: $(0,0)$ and $(\pm\sqrt{c-2k},0)$.
 
For $(0,0)$, we have
$$\Big(\frac{\partial^2 H}{\partial v\partial\phi}\Big)^2-\frac{\partial^2 H}{\partial\phi^2}\frac{\partial^2 H}{\partial v^2}=c(c-2k).$$
It is a saddle point if $c>0$ while a center point if $c<0$. Similarly, for $(\pm\sqrt{c-2k},0)$, we have
$$\Big(\frac{\partial^2 H}{\partial v\partial\phi}\Big)^2-\frac{\partial^2 H}{\partial\phi^2}\frac{\partial^2 H}{\partial v^2}=-4k(c-2k).$$
They are center points if $k>0$ and  saddle points if  $k<0$.

The critical points on the one hand may give stationary solutions and on the other hand determine local structures of level sets. To obtain a sense of the structures of the level sets, we sketch some typical level sets of $H$ for $k=-1$, $k=0$ and $k=1$. The curve $\phi^2-\frac{v^2}{3}=c$ (solid line) and $\phi^2-v^2=c$ (dashed line) are also plotted for more information and these curves are important for the patching technic in the   Subsection \ref{sec:patchedweaktravel}.

\begin{figure}[H]
\begin{center}
\includegraphics[width=0.7\textwidth]{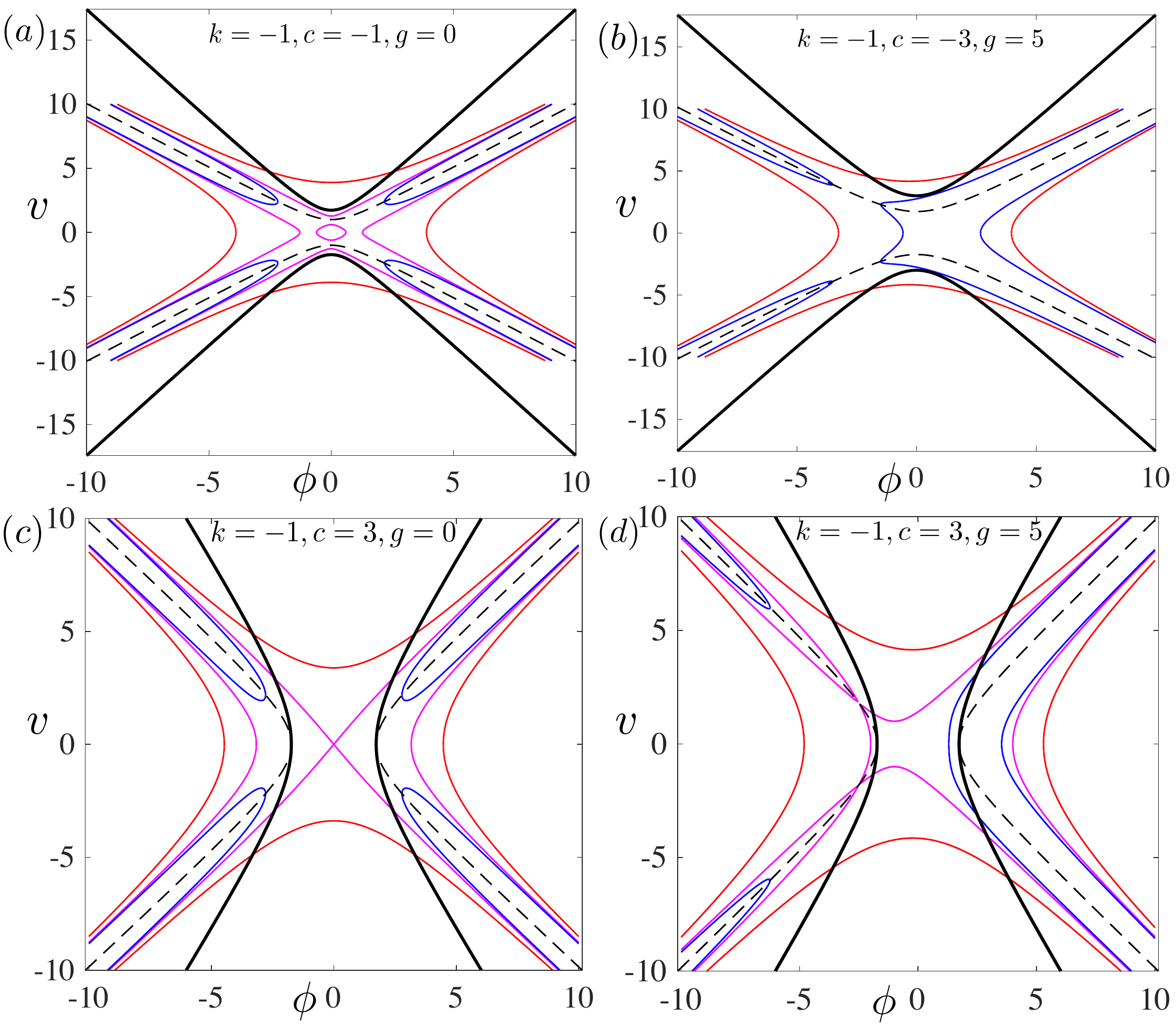}
\end{center}
\caption{Level sets of $H$ with different parameters  $c,~g$ and $h$ for $k=-1$. The black solid lines are $\phi^2-\frac{v^2}{3}=c$ and the dashed lines are $\phi^2-v^2=c$. (a) and (c) are for $g=0$ in which case the level sets are symmetric about both $\phi$ and $v$ axis. (b) and (d) are for $g>0$ where the level sets are skewed.}
\label{fig:knegative}
\end{figure}

\begin{figure}[H]
\begin{center}
\includegraphics[width=0.7\textwidth]{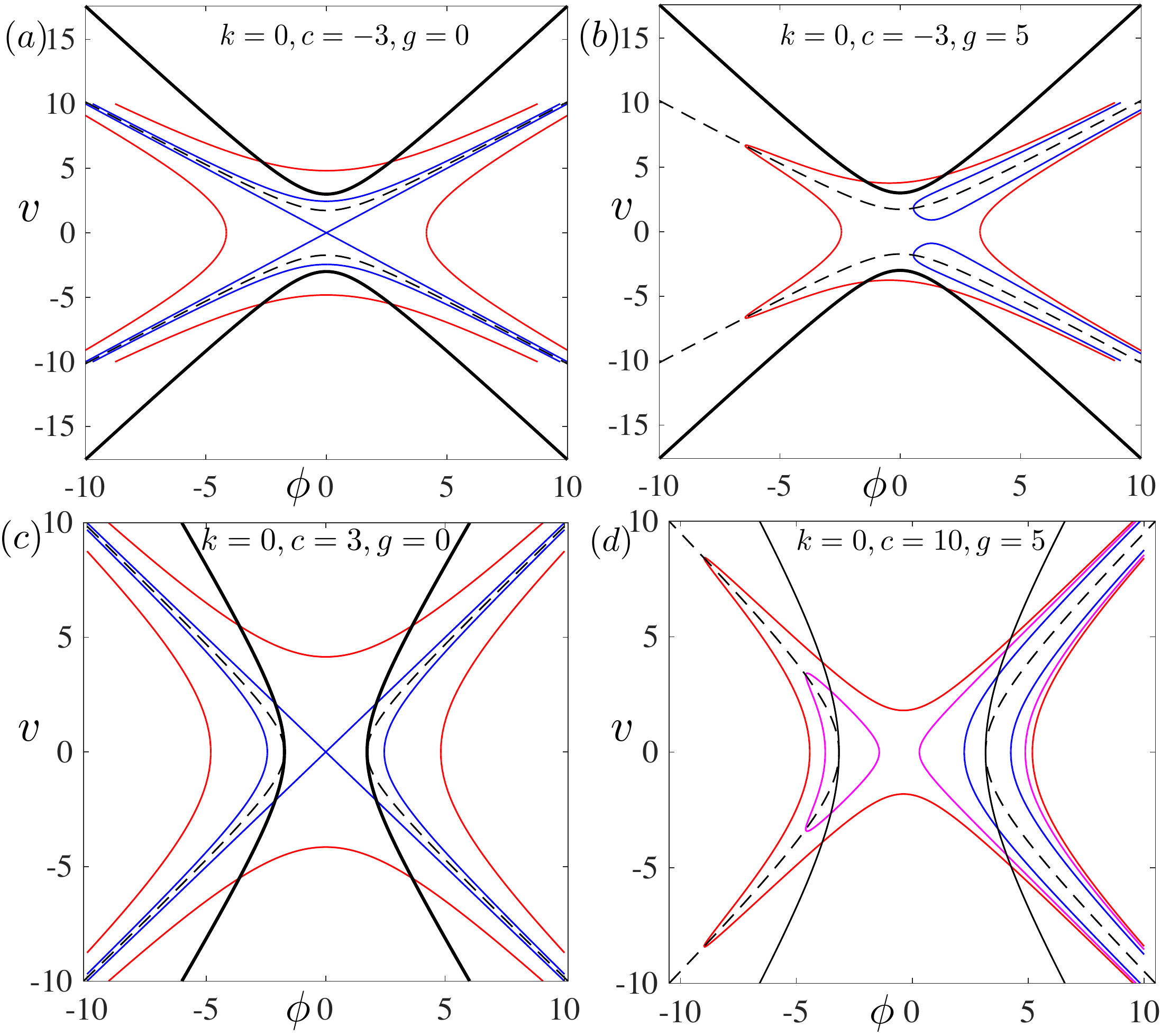}
\end{center}
\caption{Level sets of $H$ with different parameters  $c,~g$ and $h$ for $k=0$. The hyperbolas (solid lines and dashed lines) are the same as in Figure \ref{fig:knegative}.  (a) and (c) are for $g=0$, while (b) and (d) are for $g>0$. In Figure (b), there are components that are `U' shaped with both branches extending to infinity on one side of $v=0$. In Figure (d), there is a component that forms a loop. }
\label{fig:kzero}
\end{figure}
\begin{figure}[H]
\begin{center}
\includegraphics[width=0.7\textwidth]{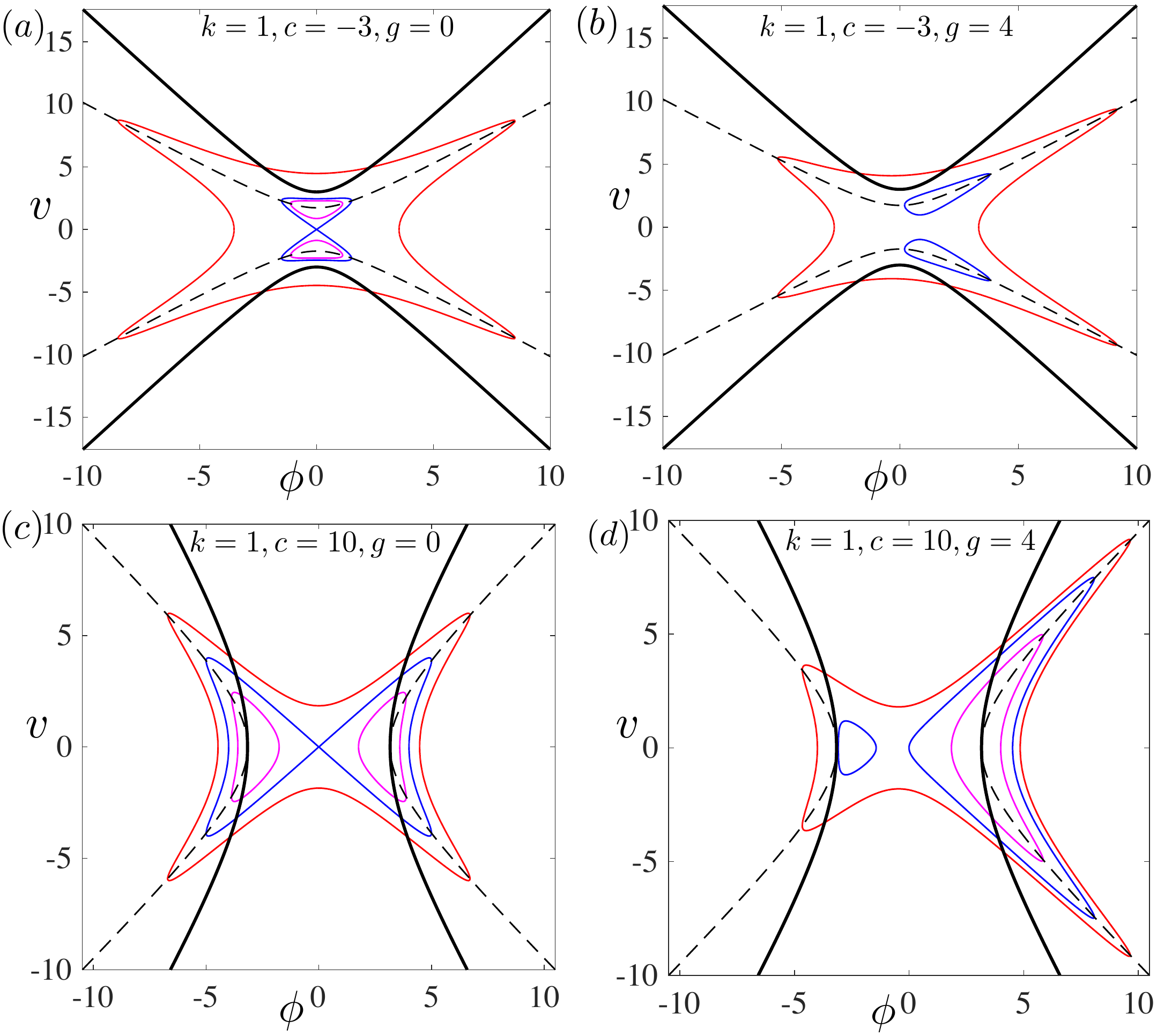}
\end{center}
\caption{Level sets of $H$ with different parameters  $c,~g$ and $h$ for $k=1$. The hyperbolas (solid lines and dashed lines) are the same as in Figure \ref{fig:knegative}.  (a) and (c) are for $g=0$, while (b) and (d) are for $g>0$. We can see that all the simple curve (a curve that does not cross itself) forms a closed loop (see Theorem \ref{thm:levelsets} for regurious proof).}
\label{fig:positive}
\end{figure}

Below, we fix parameters $(c, g, k)$ and consider those $h$ values such that the level set $H=h$ contains more than one point. 
 Denote
\begin{gather}
\Sigma\subset\{(\phi,v):H(\phi,v)=h\}
\end{gather} 
as a connected component that has more than one point. In this paper, without mentioning explicitly, the parameters $(c, g, k)$ are fixed but arbitrary. 
To relate the connected component $\Sigma$ to traveling wave solutions of the mCH equation, we first study the structure of $\Sigma$. 

Denote the upper half and lower half of the phase plane respectively by
\begin{gather}
P_+:=\{(\phi, v): v\ge 0\} ~\textrm{ and }~ P_-:=\{(\phi, v): v\le 0 \}.
\end{gather} 
We give some  lemmas that are useful to study the structure of $\Sigma$.
\begin{lemma}
Assume $\Sigma$ is a connected component of level set $H=h$ for some constant $h\in\mathbb{R}$. Then, $\Sigma$ intersects with $v=0$ for at most four points. If $\Sigma$ intersects with $v=0$, then $\Sigma$ is symmetric about $v=0$. 
\end{lemma}
The first claim follows from the fact that $H$ is a polynomial of degree four.  The proof of this lemma is simple and we omit it.

\begin{lemma}\label{lmm:arcs}
Denote $P_1:=\{(\phi, v): \phi^2-v^2-c\le 0\}$ and $P_2:=\{(\phi, v): \phi^2-v^2-c\ge 0\}$. 
Assume $(k, g)\neq (0, 0)$. Consider the two branches of $H=h$ given by
\begin{gather}\label{eq:Gammapm}
\begin{split}
&\Gamma_{+}:=\{(\phi,v):  v^2=\phi^2-c+\sqrt{c^2-4(k\phi^2-g\phi-h)}\},\\
&\Gamma_{-}:=\{(\phi,v):  v^2=\phi^2-c-\sqrt{c^2-4(k\phi^2-g\phi-h)}\}.
\end{split}
\end{gather}
Then, the following statements hold:

(i). $\Gamma_+\subset P_1$ and $\Gamma_-\subset P_2$.  

(ii). If  $\Gamma_+$ and $\Gamma_-$ are connected, then the joint is on the hyperbola $\phi^2-v^2=c$ and they  can be connected at most four points.   If  $\Gamma_{+}$   ($\Gamma_-$)  intersects $\phi^2-v^2=c$ at a point with $v\neq 0$, then it connects to $\Gamma_-$  ($\Gamma_+$).

(iii). Assume $\Gamma_{+}\cap P_{+}$ is not empty.  Then, this arc is the graph of a function $v=v(\phi)$. And  any endpoint of $\Gamma_{+}\cap P_{+}$ must be:  (a) on $v=0$,  (b) on $\phi^2-v^2=c$ with $v\neq 0$, or (c) at infinity. This statement is also true for $\Gamma_{+}\cap P_{-}$ and $\Gamma_{-}\cap P_{\pm}$.

\end{lemma}
\begin{proof}
{\bf Step 1.}
Statement (i) is obvious. 

{\bf Step 2.} For Statement (ii), assume $(\phi_0,v_0)\in \Gamma_+\cap\Gamma_-$ is a point where $\Gamma_+$ and $\Gamma_-$ are connected. By the definitions of $\Gamma_+$ and $\Gamma_-$, we have
$$c^2-4(k\phi_0^2-g\phi_0-h)=0 ~\textrm{ and }~\phi_0^2-v_0^2=c,$$
which means  the joint is on the hyperbola $\phi^2-v^2=c$.
Since $(k, g)\neq (0, 0)$, there are at most two $\phi_0$ such that $c^2-4(k\phi_0^2-g\phi_0-h)=0$. Each $\phi_0$ gives two values of $v_0$ (at most).  Hence, the two components can be connected at most four points.

Next, assume $\Gamma_{+}$  intersects $\phi^2-v^2=c$ at a point $(\phi_1,v_1)$ with $v_1\neq 0$. Hence,  $\phi_1^2-c=v_1^2>0$ and $\sqrt{c^2-4(k\phi_1^2-g\phi_1-h)}=0$, which implies $(\phi_1,v_1)\in\Gamma_-$. Similar arguments are performed if $\Gamma_-$ intersects with the hyperbola.

{\bf Step 3.}
For Statement (iii), consider the curve $\Gamma_+$ in $P_{+}$. It is the graph of function
$$v(\phi)=\sqrt{\phi^2-c+\sqrt{c^2-4(k\phi^2-g\phi-h)}}.$$
Assume  $D$ is the largest domain that $v(\phi)$ is defined on. If $\partial D$ contains some points $\tilde{\phi}$, then we have $v(\tilde{\phi})=0$ or $\sqrt{c^2-4(k\tilde{\phi}^2-g\tilde{\phi}-h)}=0$. Otherwise, there is a neighborhood such that we can extend the domain, contradicting with the fact that $D$ is the largest domain.
The same argument applies to $\Gamma_+\cap P_-$ and $\Gamma_-\cap  P_{\pm}$. 
\end{proof}

Now, we are able to conclude the structure of  the connected component $\Sigma$.
\begin{proposition}\label{pro:structureGam}
 Suppose $\Sigma$ is a connected component of a level set of $H=h$ that does not contain $(\phi_0, 0)$ with $\phi_0^2=c$. Then,  one of the following two statements holds:

(i). $\Sigma=\Gamma$ and $\Gamma$ is a simple curve that is either a closed loop, or with both sides  tending to infinity in the phase plane.  

(ii). $\Sigma=\cup_{i=1}^N\Gamma_i$, with $\Gamma_i$ being a simple curve as described in (i). If $\Gamma_i\cap \Gamma_j\neq \emptyset$ for $i\neq j$,  then any $p\in \Gamma_i\cap \Gamma_j$ is a critical point of System \eqref{eq:system2}.

If the two sides of $\Gamma_i$ tend to infinity,  then both $|v|$ and $|\phi|$ tend to infinity. Moreover,  for $\Sigma\cap P_+$  ($\Sigma\cap P_-$) containing more than one point, it can be represented by a graph of some function $v=f(\phi)$ if and only if $\Sigma$ does not intersect $\phi^2-v^2=c$.
\end{proposition}

The claims in Proposition \ref{pro:structureGam} are clear by Lemma \ref{lmm:arcs}, but the complete proof is tedious and we put it into Appendix \ref{app:structure}.
\begin{remark}
By using the Poincare-Bendixson Theorem, we can conclude that if $\Sigma$ is bounded and does not pass through a critical point, it must be a simple loop. When $k>0$, from \eqref{eq:Gammapm} we can tell that the level set of $H=h$ is bounded for any $h\in\mathbb{R}$. Hence, the connected component of level sets $H=h$ when $k>0$ is a simple loop (see Figure \ref{fig:positive}), if it contains more than one point and does not pass through a critical point.
\end{remark}

\subsection{Traveling wave solutions}\label{subsec:travelsolution}

From here on, we will use the same notation $\Gamma$ to represent a simple curve in the connected component $\Sigma$ as described in Proposition \ref{pro:structureGam} that does not pass through a point $(\phi_0, 0)$ with $\phi_0^2=c$. Due to $v=\frac{d\phi}{d\xi}$, we have
\begin{gather}
d\xi=\frac{1}{v}d\phi.
\end{gather}
To be convenient, we consider those simple curves $\Gamma$ for which we can define a period
\begin{gather}\label{eq:period}
T:=\int_{\Gamma}\frac{1}{v}d\phi,
\end{gather}
where the orientation of $\Gamma$ is chosen such that $T\in [0,\infty]$. Note that it is possible that $T=\infty$.

\begin{remark}
Requiring that $T$ is defined excludes some $U$-shaped curves, as shown in Figures \ref{fig:knegative}(b) and \ref{fig:kzero}(b).  $T$ can not be defined for such a curve with any orientation. For example, consider the $U$-shaped curve $\Gamma$ in $P_+$ from Figure \ref{fig:kzero}(b), with clockwise orientation.
$\Gamma=\gamma_1\cup\gamma_2$, $\gamma_1\subset\Gamma_+$ and $\gamma_2\subset\Gamma_-$.  It is clear that $\int_{\gamma_1}\frac{1}{v}d\phi=+\infty$ and $\int_{\gamma_2}\frac{1}{v}d\phi=-\infty$. Hence, $T=\int_{\gamma_1\cup\gamma_2}\frac{1}{v}d\phi$ is not defined on such a curve.
\end{remark}

There are some key points on $\Gamma$ that need careful investigation. For a point $(\phi, v)$, we denote
\begin{gather}
\Gamma_{\delta}(\phi,v):=\Gamma\cap B((\phi,v), \delta), ~\delta>0,
\end{gather}
with the same orientation as $\Gamma$.  The following two lemmas tell us the behaviors of the solutions given by $\Gamma$ near some key points:
\begin{lemma}\label{lmm:timenearzero}
Suppose $\Gamma$ passes through $(\phi_*, 0)$ and $(\phi_*, 0)$ is not a critical point, then $\Gamma$ is perpendicular with $v=0$ axis at $(\phi_*, 0)$ and for $\delta>0$ small enough, we have 
\begin{gather*}
\int_{\Gamma_{\delta}(\phi_*,0)}\frac{1}{v}d\phi<\infty.
\end{gather*}
\end{lemma}
\begin{proof}
Suppose $(\phi_*, 0)$ is such a point. Note that $\nabla H(\phi_*,0)=(\partial_\phi H,\partial_vH)\Big|_{(\phi_*,0)}\neq 0$ since there is no critical point on $\Gamma$. At $(\phi, v)=(\phi_*, 0)$, we have
$$\partial_v H(\phi_*,0)=-v(\phi^2-v^2)+cv\Big|_{(\phi_*,0)}=0.$$ 
Hence, $\partial_{\phi}H\neq 0$ and by the implicit function theorem we can solve $\phi$ in terms of $v$ near $(\phi_*, 0)$ as $\phi=\phi(v)$. Hence,  $\phi(0)=\phi_*$,          
$\phi'(0)=0$.   Taylor expansion implies  that $1/v\sim |\phi-\phi_*|^{-1/p}$ where $p\ge 2$. Hence, $\int \frac{1}{v}d\phi$ is integrable near this point. 
\end{proof}

There may be some trajectories that connect saddles. We have the following lemma.
\begin{lemma}\label{lmm:timenearpt}
Suppose $(\phi_*, v_*)\in \Gamma$ is a non-degenerate critical point of \eqref{eq:system2}.  
If $v_*\neq 0$, then the critical point is on $\phi^2-v^2=c$ and there exists $\delta>0$ such that
\begin{gather*}
\int_{\Gamma_{\delta}(\phi_*, v_*)}\frac{1}{v}d\phi<\infty.
\end{gather*}
If $v_*=0$ and $\phi_*^2\neq c$,  then for any $\delta>0$
\begin{gather*}
\left|\int_{\Gamma_{\delta}(\phi_*, v_*)\cap \{v\ge 0\}}\frac{1}{v}d\phi\right|=\left|\int_{\Gamma_{\delta}(\phi_*, v_*)\cap \{v\le 0\}}\frac{1}{v}d\phi\right|=\infty.
\end{gather*}
\end{lemma}
\begin{proof}
Note that $(\phi_*, v_*)$ must be a saddle point.
If $v_*\neq 0$, from \eqref{eq:system2} we can see that this critical point must be on the curve $\phi^2-v^2-c=0$. 
It is clear that $\int\frac{1}{v}d\phi$ is integrable near the critical point. 

If $v_*=0$ and $\phi_*^2-v_*^2-c\neq 0$,  then $(\phi_*, v_*)$ is also a critical point of System \eqref{eq:system}.  The existence and uniqueness theorem of System \eqref{eq:system} holds in the neighborhood of this critical point. 
As a result, $\Delta\xi=\int\frac{1}{v}d\phi$ must be infinite near this critical point because the critical point corresponds to a stationary solution of System \eqref{eq:system}.
\end{proof}

Now, let us define the multi-valued solution to the mCH equation for our convenience of discussion here.  
\begin{definition}\label{def:multiweak}
For any $a,b\in\mathbb{R}$, denote $J=(a, b)$. We say $sgn(J):=sgn(b-a)$ and $x\in J$ if $x\in (\min(a,b),\max(a,b))$.  For a function $u$ defined on $(\min(a,b),\max(a,b))$, we also say it is defined on $J$ and  define the integral on $J$ as
$$\int_J u d\xi=sgn(J)\int_{\min(a,b)}^{\max(a,b)} ud\xi.$$
\end{definition}
\begin{definition}\label{def:multi}
We say $u:\mathbb{R}\times[0, T]  \to 2^{\mathbb{R}}, (x, t)\mapsto u(x, t)$ is a multi-valued solution to $\eqref{mCH}$ subject to initial data $u(x,0)=u_0(x)$ if there exist signed intervals $J_n(t)=(a_n(t), b_n(t)),~ n\in\mathbb{Z},$ (see Figure \ref{fig:multiweak}) and a single-valued function $u_n(\cdot, t)$ defined on each $J_n(t)$ with $\gamma_n(t)=\{(x, u_n(x, t)), x\in J_n(t) \}$  such that 
\begin{itemize}
\item $a_n, b_n\in C^1[0, T]$, $b_n(t)=a_{n+1}(t)$, $\lim_{n\to-\infty}a_n(t)=-\infty$, $\lim_{n\to+\infty}b_n(t)=+\infty$
\item $\gamma_n$ are non-intersecting and $\forall t\in [0,T]$, $\cup_n \bar{\gamma}_n=\{(x, u): x\in\mathbb{R}, u\in u(x, t)\}$.
\item $\|u_n\|_{W^{1,\infty}(J_n(t))}$ is bounded as functions of $t$.  For any $ [a, b]\times [t_1, t_2]\subset \{(x, t): x\in J_n(t)\}$, it holds $u_n\in C(t_1, t_2; H^1([a, b]))$. $u_n(b_n(t),t)=u_{n+1}(a_{n+1}(t),t)$. 
\item  It holds that $\forall \varphi\in C_c^{\infty}([0,T),C_c^{\infty}(\mathbb{R}))$:
\begin{multline*}
\int_0^T\sum_n \int_{J_n} u(\varphi_t-\varphi_{xxt}) dxdt
+\int_0^T\sum_n\int_{J_n} (2ku+u^3+uu_x^2)\varphi_xdxdt\\
-\frac{1}{3}\int_0^T\sum_n\int_{J_n}u^3\varphi_{xxx}dxdt-\frac{1}{3}\int_0^T\sum_n\int_{J_n} u_x^3\varphi_{xx}dxdt\\
=-\sum_{n}\int_{J_n(0)} u_0 \varphi(x, 0)dx+\sum_n\int_{J_n(0)} u_0\varphi_{xx}(x,0)dx.
\end{multline*}
\end{itemize}
 \end{definition}

\begin{figure}[H]
\begin{center}
\includegraphics[width=0.5\textwidth]{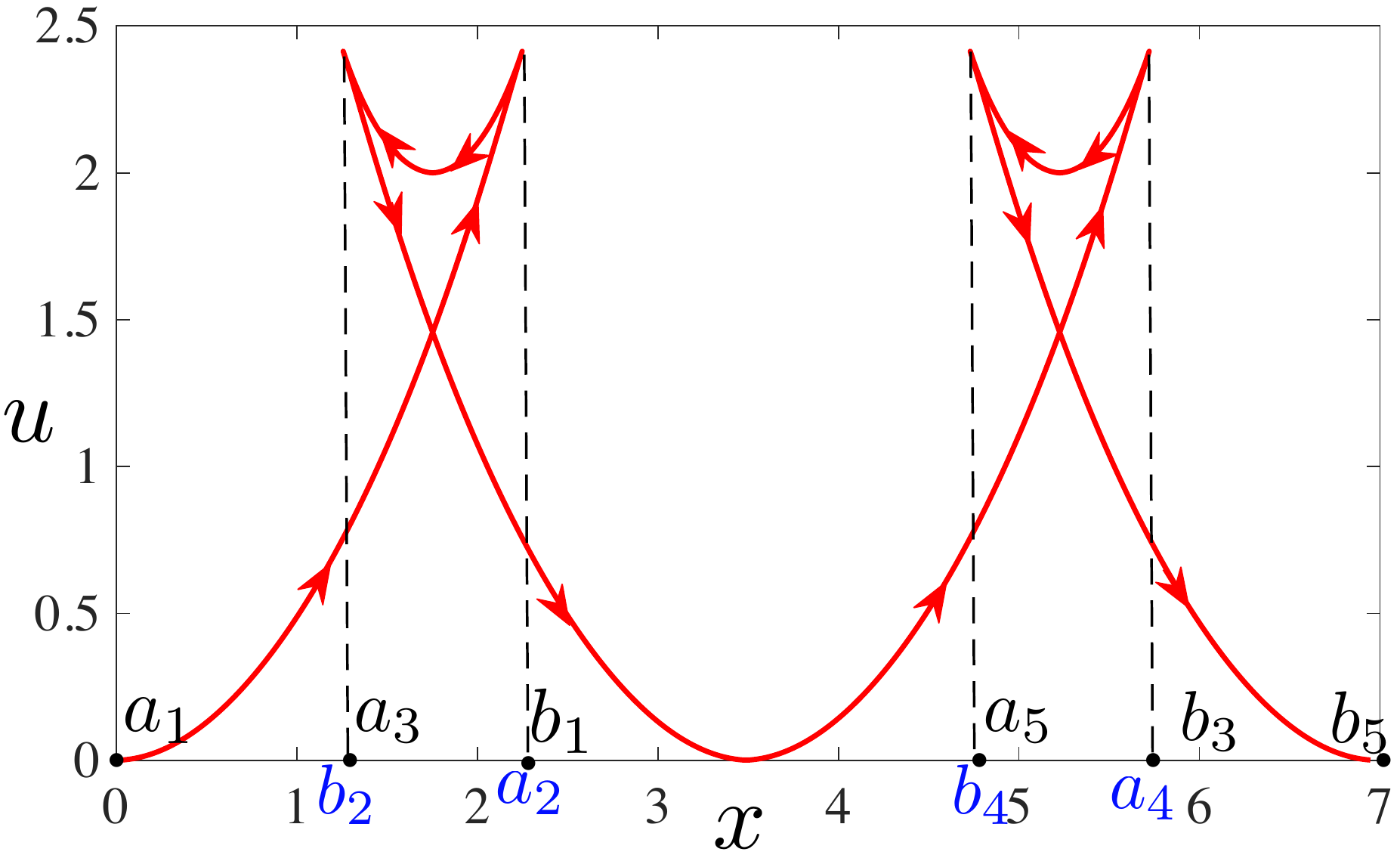}
\end{center}
\caption{A typical multi-valued solution $u(x,t)$ (Definition \ref{def:multi}) at a fixed time $t$. The intervals $(a_n, b_n) (n=1,2,\ldots,5)$ are demonstrated on the $x$-axis. In the picture, $b_1>a_1$, $b_2<a_2=b_1$ and $b_3>a_3=b_2$ etc..}
\label{fig:multiweak}
\end{figure}
 
With similar computations as in the proof of Theorem \ref{thm:jumpcondition}, we find
\begin{lemma}\label{lmm:multiJump}
Use the same notations as in Definition \ref{def:multi} and assume that $t\mapsto u_i(\cdot, t)$ solves the mCH equation on the region $\{(x, t): x\in J_i(t)=(a_i(t),b_i(t))\}$ classically for $i\in\mathbb{Z}$. Let $m_i(x,t)=u_i(x,t)-\partial_{xx}u_{i}(x,t)$ for $x\in J_i(t)$ and  $x_i(t):=a_i(t)$. If the jump conditions \eqref{eq:jumpcd1} or \eqref{eq:continu} are satisfied where 
\begin{align*}
v_r^i(t):=\lim_{\substack{x\in J_i,\\x\rightarrow a_i(t)}}\partial_xu_i(x,t),~~v_\ell^i(t):=\lim_{\substack{x\in J_{i-1},\\x\rightarrow b_{i-1}(t)}}\partial_xu_{i-1}(x,t)
\end{align*}
and
\begin{align*}
A^i_r:=\lim_{\substack{x\in J_i,\\x\rightarrow a_i(t)}}m_{i}\big[u^2_{i}-(\partial_xu_i)^2-a'_i(t)\big],~~A^i_r:=\lim_{\substack{x\in J_{i-1},\\x\rightarrow b_{i-1}(t)}}m_{i-1}\big[u^2_{i-1}-(\partial_xu_{i-1})^2-a'_i(t)\big],
\end{align*}
 then  $u$ is a multi-valued solution as in Definition \ref{def:multi}
 \end{lemma}

We now relate the simple curves $\Gamma$ to traveling wave solutions of the mCH equation:
\begin{proposition}\label{pro:leveltoSol}
Suppose $\Gamma$ is a simple curve described in Proposition \ref{pro:structureGam} that does not contain $(\phi_0, 0)$ with $\phi_0^2=c$. Let $T$ be defined by \eqref{eq:period} and we assume $T>0$.

(i).
If $\Gamma$ does not pass any critical point on $v=0$, then $\Gamma$ corresponds to a traveling wave solution $u(x,t)=\phi(\xi)$ to the mCH equation $\eqref{mCH}$, where $\xi=x-ct$.  And if  $\Gamma$ does not cross $\phi^2-v^2=c$, then $u$ is a single-valued solution. Moreover, $u$ is a periodic traveling wave solution if and only if $\Gamma$ is a closed loop. 

(ii).
If there are critical points of the form $(\phi_*, 0)$ dividing $\Gamma$ into several segments as $\Gamma=\cup_i\Gamma_i$ (see Figure \ref{fig:prostatement2}), then each $\Gamma_i$ corresponds to a non-periodic traveling wave solution $u(x,t)=\phi(\xi)$ for $-\infty<\xi<\infty$ where $\xi=x-ct$.
\end{proposition}

\begin{figure}[H]
\begin{center}
\includegraphics[width=0.5\textwidth]{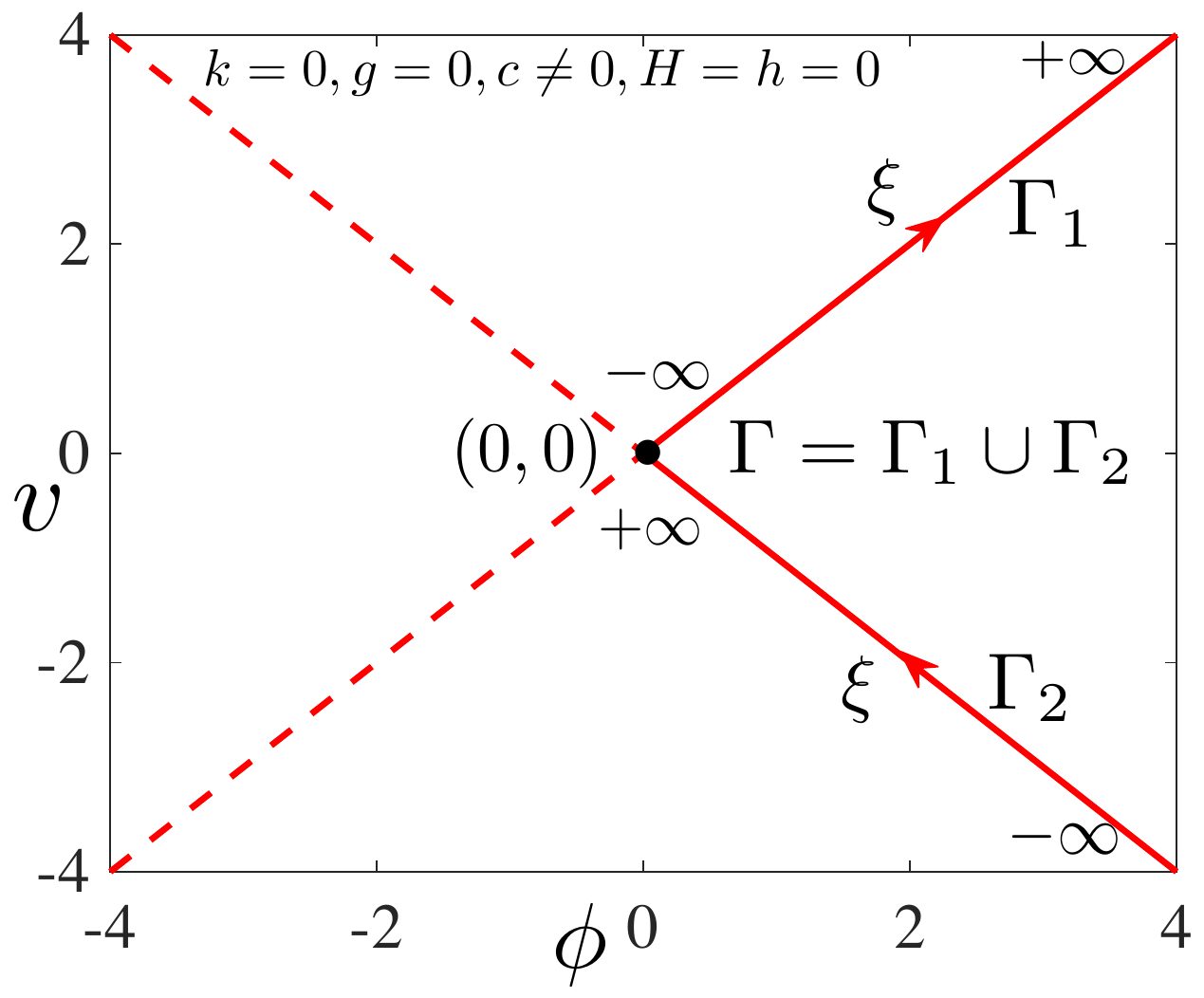}
\end{center}
\caption{The curves (both dashed lines and solid lines together) in this figure make up a connected component ($\phi^2-v^2=0$) of level set $H=0$ for $g=k=0$ and $c\neq0$. $\Gamma=\Gamma_1\cup\Gamma_2$ is a simple curve as stated in Statement (ii) of Proposition \ref{pro:leveltoSol}. $(\phi_*, 0)=(0,0)$ is a saddle point that divides $\Gamma$ into $\Gamma_1$ and $\Gamma_2$.  $\Gamma_1$ gives a solution $\phi_1(\xi)=A_1e^{\xi}$ to the mCH equation and $\Gamma_2$ gives $\phi_2(\xi)=A_2e^{-\xi}$ for some constants $A_i>0,~i=1,2$ and $\xi\in\mathbb{R}$.} 
\label{fig:prostatement2}
\end{figure}

\begin{proof}
Proof for Statement (i). 

{\bf Step 1.} Preparations.

We assume first that $\Gamma$ does not pass through any critical point on $v=0$.  Fix $A_0=(\phi_0, v_0)\in\Gamma$, and use $\Gamma(A_0, (\phi, v))$ to mean the subcurve of $\Gamma$ from $A_0$ to $(\phi, v)$ with the same orientation as when we define $T$ in Equation \eqref{eq:period}. We define
\begin{gather}\label{eq:sol}
\xi=\int_{\Gamma(A_0,(\phi, v))} \frac{1}{v}d\phi.
\end{gather}
In the case $T\in (0,\infty)$ (or $\Gamma$ is a loop), this expression is defined for all $\xi\in \mathbb{R}$ if we orbit around the loop repeatedly.
This integral then gives a set value mapping $\phi: \mathbb{R}\to 2^{\mathbb{R}}$, $\xi\mapsto \phi(\xi)$ provided $T>0$.

Assume $\phi^2-v^2=c$ divides $\Gamma$ into several arcs (may be one segment). From Statement (ii) of Lemma \ref{lmm:arcs}, we know there are at most five such arcs with four joint points. Moreover, each arc is contained in  $\Gamma_+$ or $\Gamma_-$. Hence, $d\xi$ has a definite sign in each arc. If the sign of $d\xi=\frac{1}{v}d\phi$ are different on two arcs, we get a multi-valued solution, which happens if one arc is below strictly under the other in $P_+$ or $P_-$. 

The functions $(\phi(\xi), v(\xi))$ solves System \eqref{eq:system} on each arc and therefore $u(x,t)$ solves the mCH equation in the classical sense, where
$$u(x,t)=\phi(\xi)~\textrm{ and }~x-ct=\xi.$$

{\bf Step 2.} For the case when $\Gamma\cap \{(\phi,v):\phi^2-v^2=c\}= \emptyset$, $d\xi$ has a definite sign and we have a single-valued function.

{\bf Step 3.} The case for $\Gamma\cap \{(\phi,v):\phi^2-v^2=c\}\neq \emptyset.$

In this case, we may have a multi-valued function. We only have to check the jump conditions across the hyperbola $\phi^2-v^2=c$. 
Assume $A=(\phi_A,v_A)\in\Gamma\cap \{(\phi,v):\phi^2-v^2=c\}$. The jump for $u(x,t)$ happens at
$$x-ct=\xi_A=\int_{\Gamma(A_0,A)} \frac{1}{v}d\phi.$$
Hence, the jump line in $x,t$ plane is given by $x_A(t)=ct+\xi_A.$ We have
$$v_{\ell}=v_r=v_A~\textrm{ and }~\frac{d}{dt}x_A(t)=c,$$
which agrees with the wave speed.
Notice that $m=u-u_{xx}=\phi-\phi''=\phi-\frac{dv}{d\xi}$. By System \eqref{eq:system}, we have 
$$\lim_{\xi\neq \xi_A,\xi\to \xi_A}m(u^2-\partial_x u^2-x'_A(t))=\lim_{\xi\neq \xi_A,\xi\to \xi_A}(\phi-\phi'')(\phi^2-v^2-c)=g-2k\phi_A.$$
The  jump condition \eqref{eq:continu} is verified. Hence, $u(x,t)$ is a weak solution to the mCH equation.

{\bf Step 4.} Periodicity.

If $\Gamma$ is a closed loop,  we have $0<T<\infty$. The solution is therefore periodic traveling wave solution. 

When $\Gamma$ is not a closed loop, both sides of $\Gamma$ tend to infinity. By definition of $\Gamma_{\pm}$ \eqref{eq:Gammapm}, we have $|v|\sim |\phi|$ as $|\phi|\to\infty$. This means $\xi$ extends to $-\infty$ on one side and extends to $+\infty$ on the other side by formula \eqref{eq:sol}. This yields a solution defined on $\mathbb{R}$ which is not a periodic solution.

Proof for Statement (ii). 

Now assume there are critical points of the type $(\phi_*, 0)$ on $\Gamma$. If $\Gamma$ is a loop, both ends of a segment are such critical points. If $\Gamma$ is an open curve extending to infinity, then the endpoints of each segment are either such critical points or infinity. By Lemma \ref{lmm:timenearpt}, the integrals at the two ends of such a segment are infinity. Hence, the function similarly given by \eqref{eq:sol} is defined on $\mathbb{R}$. That it is a solution can be verified similarly as the first case.
\end{proof}

Next, we give a theorem to describe the typical traveling wave solutions to the mCH equation \eqref{mCH} given by level sets of $H$. When $k>0$, the typical solutions are multi-valued periodic traveling wave solutions.  And when $k\leq0$,  the typical solutions are unbounded and single-valued, defined on $\mathbb{R}$. 
\begin{figure}[H]
\begin{center}
\includegraphics[width=0.8\textwidth]{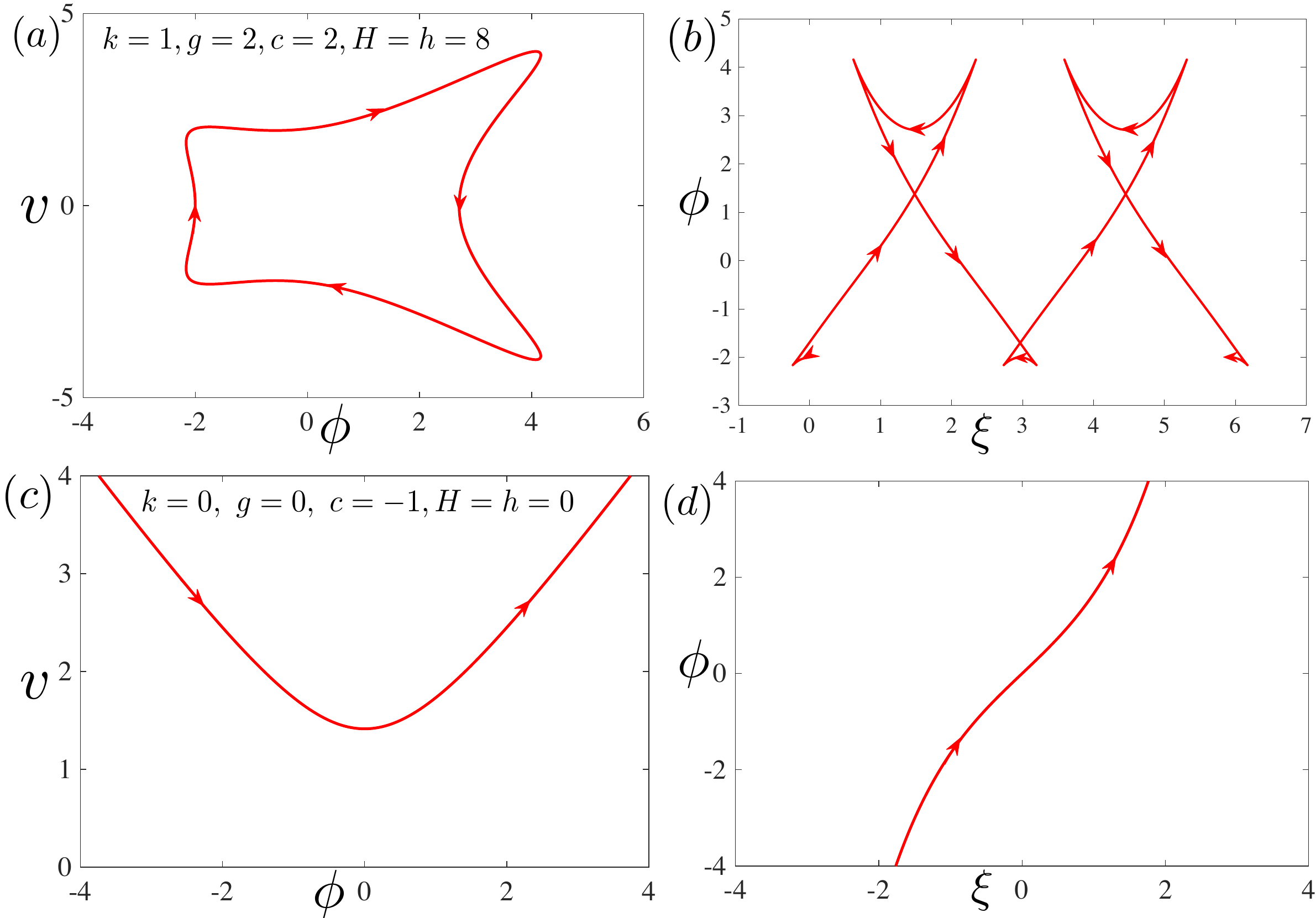}
\end{center}
\caption{Illustration of Theorem \ref{thm:levelsets}. Typical simple curves $\Gamma$ and the corresponding traveling wave solutions . A closed Loop  when $k>0$ in phase plane (a) corresponds to a multi-valued periodic traveling wave solution in (b). A simple curve when $k\leq 0$ in phase plane (c) corresponds to a traveling wave solution (d).  (a), (b) are for (i) in Theorem \ref{thm:levelsets} and (c), (d) are for (ii) in Theorem \ref{thm:levelsets}.  }
\label{fig:typicalweak}
\end{figure}
 
 We have the following theorem.
\begin{theorem}\label{thm:levelsets}
Consider System \eqref{eq:system2}. Assume $S=\{(\phi,v):H(\phi,v)=h\}$ is nonempty and does not pass through any critical point. 

(i). If $k>0$ for any traveling speed $c\in\mathbb{R}$, each simple curve in $S$ is a loop in the phase plane of System \eqref{eq:system2}.  There exists $h_0>0$ such that  for $c<0$ or  $c\ge 0$, $h>h_0$,  there is a loop $\Gamma\subset S$ that gives a multi-valued traveling wave solution of the mCH equation. If $T>0$, this solution is a periodic solution on $\mathbb{R}$. (See Figure \ref{fig:typicalweak} (a) and (b).) 

(ii). If $k\le 0$, for any traveling speed $c\in\mathbb{R}$, there exists $h_0$ such that when $h>h_0$,  there is a simple curve $\Gamma\subset S$  with both sides tending to infinity which corresponds to a single-valued traveling wave solution to the mCH equation globally on $\mathbb{R}$ and $\lim_{|\xi|\to\infty}|\phi(\xi)|=+\infty$. (See Figure \ref{fig:typicalweak} (c) and (d).) 
\end{theorem}

\begin{proof}
$\Gamma_\pm$ are given by \eqref{eq:Gammapm}.
Set 
$$f(\phi):=c^2-4(k\phi^2-g\phi-h)=c^2-4k\Big(\phi-\frac{g}{2k}\Big)^2+\frac{g^2}{k}+4h.$$
If $S$ is nonempty, there exists $\phi\in\mathbb{R}$ such that
$$f(\phi)\geq0~\textrm{ and }~\phi^2-c+ \sqrt{f(\phi)}\geq0.$$

(i). For $k>0$,   any level set is bounded. 
Indeed,   $\phi$ must be bounded to ensure $f(\phi)\ge 0$. Hence, $(\phi,v)$ is defined in a bounded set of the phase plane.

By Proposition \ref{pro:structureGam}, for each simple curve inside a connected component of $S$ must be a loop. When $c<0$ or $h>\frac{4kc-c^2-4g\sqrt{c}}{4}~(c\ge 0)$, $\Gamma_+$ intersects with $\phi^2-v^2=c$ at a point with $v\neq 0$. Choose $\Gamma$ to be the simple curve that passes through this intersection. By Proposition \ref{pro:leveltoSol}, $\Gamma$ gives a traveling wave solution. If $T>0$, it is clearly a periodic traveling wave solution on $\mathbb{R}$. By Lemma \ref{lmm:arcs}, $\Gamma_-\cap \Gamma\neq \emptyset$. Since $\Gamma$ does not pass through any critical point, $\Gamma_-$ is below $\Gamma_+$ and the solution is multi-valued by the proof of Proposition  \ref{pro:leveltoSol}.

(ii). Consider that $k=0$. We find that 
\[
v^2=\phi^2-(c\pm \sqrt{c^2+4(g\phi+h)}).
\]

When $g=0$, if we choose $h$ such that  $c^2+4h>0$. $\Gamma_+$ and $\Gamma_-$ do not intersect. Also, there is a connected component in $\Gamma_+$ above $v=0$ which is a function graph and extends to infinity. Since $d\xi=\frac{1}{v}d\phi>0$ on this component.  $T$ is defined and $T\ge \infty$. Hence it is a solution defined on $\mathbb{R}$

When $g>0$, consider $h>\frac{(|c|-c)^2-(c^2+4g\sqrt{|c|})}{4}$. Then, 
$$v^2=\phi^2-c-\sqrt{c^2+4g\phi+4h}=0$$ 
has a root at $\phi_0$ satisfying $\phi_0^2-c>0$ and $\phi^2-c-\sqrt{c^2+4g\phi+4h}>0$ for all $\phi>\phi_0$. 
$$\Gamma:=\Big\{(\phi,v):\phi\geq\phi_0,v^2=\phi^2-c-\sqrt{c^2+4g\phi+4h}\Big\}$$
yields a single-valued solution defined on $\mathbb{R}$ and $\lim_{|\xi|\to\infty}|\phi(\xi)|=\infty$ by Proposition \ref{pro:leveltoSol}.

Consider $k<0$, it is clear that when $h>-\frac{g^2}{4k}$, the term in the square root is always positive and $\Gamma_+$,
\begin{align*}
v^2=\phi^2-c+\sqrt{c^2-4(k\phi^2-g\phi-h)},
\end{align*}
 does not cross $\phi^2-v^2=c$ and $\phi^2-c+\sqrt{c^2-4(k\phi^2-g\phi-h)}>0$ for all $\phi$. $\Gamma=\Gamma_+\cap P_+$ (or $\Gamma=\Gamma_+\cap P_-$) is a simple curve as described in Proposition \ref{pro:structureGam}, which gives a single-valued solution defined on $\mathbb{R}$ and $\lim_{\xi\to\infty}|\phi(\xi)|=\infty$ by Proposition \ref{pro:leveltoSol}.
\end{proof}

Besides the typical solutions mentioned in Theorem \ref{thm:levelsets}, there are also some other different kinds of solutions. Let us make some complementary discussions. 
\begin{enumerate}
\item For $k>0$, the typical solutions obtained in Theorem \ref{thm:levelsets} are  multi-valued periodic traveling wave solutions (see Figure \ref{fig:typicalweak} (b)). 
There are also single-valued periodic traveling wave solutions to the mCH equation.
For example, if $k=1$, $g=2$, $c=4$, the level set $H=-4.1$ does not intersect $\phi^2-v^2=c$ and it gives a classical periodic solitary wave solution to the mCH equation. 
\item  For $k\leq 0$, there may exist a simple curve $\Gamma$ that is a loop, which gives a periodic solution to the mCH equation.  See Figure \ref{fig:knegative} (a) (single-valued) and Figure \ref{fig:kzero} (d) (multi-valued). 

\item The special case $k=0$ and $g=0$ is very important for our discussion in Section \ref{sec:peakonsolutions}.
\begin{align}
H=\frac{1}{4}(\phi^2-v^2)^2-\frac{1}{2}c(\phi^2-v^2).\label{eq:k0g0}
\end{align}
The level sets are given by
\[
\phi^2-v^2=A,~~ \textrm{ for some }~~A\in\mathbb{R}.
\]
If $A<0$, the two curves of $\phi^2-v^2=A$ corresponding to two strictly monotonic solutions to the mCH equation (see Figure \ref{fig:typicalweak} (c),(d)):
$$\phi(\xi)=\pm\sqrt{-A}\sinh(\xi+d)~\textrm{ for some constant }~d\in\mathbb{R}.$$
If $A>0$, we solve that 
\begin{gather*}
\phi(\xi)=\pm\sqrt{A}\cosh(\xi+d)~\textrm{ for some constant }~d\in\mathbb{R}.
\end{gather*}
The case $A=0$ is very interesting as it passes through the critical point $(0, 0)$ which is a saddle point when the traveling speed $c\neq 0$. 
The time spent near $(\phi,v)=(0,0)$ now is infinity.  Later in Section \ref{sec:peakonsolutions}, we obtain peakon weak solutions (see Remark \ref{remark:specialcase}) by the patching technic.
\end{enumerate}

\subsection{Patched  traveling peakon weak solutions}\label{sec:patchedweaktravel}

By Theorem \ref{thm:levelsets}, the typical solutions  for $k>0$ are multi-valued periodic solutions while the typical solutions for $k\le 0$ are single-valued solutions which tend to infinity  as $|\xi|\to\infty$.  In this subsection, we introduction the patching technic to obtain patched bounded single-valued weak solutions from these two types of solutions. 

First, we present a patching criterion to construct a patched periodic weak solution by connecting smooth solutions given by  level sets of $H$.

\begin{figure}[H]
\begin{center}
\includegraphics[width=0.7\textwidth]{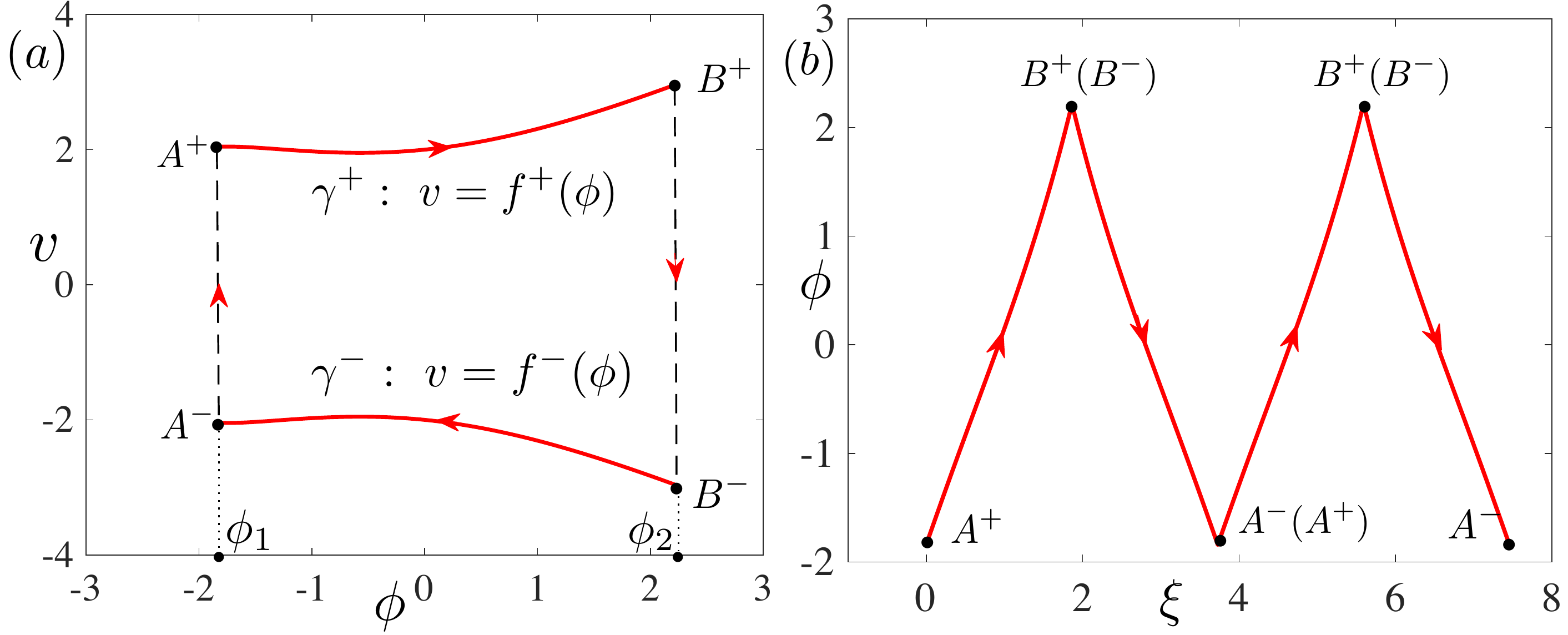}
\end{center}
\caption{A sketch to illustrate  Proposition \ref{pro:patching1}.  } 
\label{fig:proposition3.4}
\end{figure}
\begin{proposition}\label{pro:patching1}
Fix any $k\in\mathbb{R}, c\in\mathbb{R}$ and $g\ge 0$. Suppose  curves
\begin{gather*}
\gamma^+\subset \{(\phi, v): H=h,~v\geq0\},~\gamma^-\subset \{(\phi, v): H=h',~v\leq0\},\textrm{ for some }h\textrm{ and }h',
\end{gather*}
are function graphs of two continuous functions $v=f^+(\phi)$, $v=f^-(\phi)$ for $\phi\in [\phi_1, \phi_2]$ ($\phi_1<\phi_2$) (see Figure \ref{fig:proposition3.4} (a)).  Denote points $A^+: (\phi_1, v_1^+)=(\phi_1, f^+(\phi_1))$ and $B^+:(\phi_2, v_2^+)=(\phi_2, f^+(\phi_2))$, $A^-:(\phi_1, v_1^-)=(\phi_1, f^-(\phi_1))$ and $B^-:(\phi_2, v_2^-)=(\phi_2, f^-(\phi_2))$. 

If $\gamma^+$ and $\gamma^-$ do not pass through critical points and the end points $A^+,B^+,A^-$ and $B^-$ satisfy:
\begin{gather}\label{eq:conditionsforcat}
(\phi_i)^2-\frac{1}{3}\Big((v_i^+)^2+v_i^+v_i^-+(v_i^-)^2\Big)=c,~i=1,2,
\end{gather}
then $\Gamma=\gamma^+\cup B^+B^-\cup\gamma^-\cup A^-A^+$ gives a patched periodic traveling peakon weak solution with wave speed $c$ (see Figure \ref{fig:proposition3.4} (b)), where $PQ$ represents the line segment between two points $P$ and $Q$. 
\end{proposition}
\begin{proof}

Let $\Gamma$ be oriented clockwise. We use $\Gamma(A^+, (\phi, v))$ to mean the part of $\Gamma$ from $A^+$ to $(\phi, v)$ with the same orientation as $\Gamma$. The relation
\begin{gather}
\xi=\int_{\Gamma(A^+,(\phi, v))} \frac{1}{v}d\phi
\end{gather}
then gives a mapping $\xi\mapsto \phi(\xi)$. First note that $\int_{A^-A^+}\frac{1}{v}d\phi=0=\int_{B^+B^-}\frac{1}{v}d\phi$. On $\gamma^+$,  since $1/v>0$ and $d\phi>0$, we have $d\xi>0$. Similarly, $d\xi>0$ on $\gamma^-$.  Hence, $\phi(\xi)$ is a single-valued mapping.
Since the arcs do not pass trough critical points, $T$ defined by \eqref{eq:period} satisfies
\[
T=\int_{\Gamma}\frac{1}{v}d\phi\in (0,\infty).
\]
Hence, $\phi(\xi)$ is a periodic function. Denote $T_1=\int_{\gamma^+}\frac{1}{v}d\phi$ and consider 
\begin{align*}
u(x, t)=
\begin{cases}
u_1(x-ct)=\phi(x-ct),~~x-ct=\xi\in(0,T_1);\\
u_2(x-ct)=\phi(x-ct),~~x-ct=\xi\in(T_1,T),
\end{cases}
\end{align*}
and $m_i=u_i-\partial_{xx}u_i$ for $i=1,2.$
 The functions $u_1$ and $u_2$  solve the mCH equation in the classical sense respectively for $\xi\in (0,T_1)$ and $\xi\in (T_1,T)$. We only have to verify that the jump conditions are satisfied at the joints $A^+A^-$ and $B^-B^+$, corresponding to $x-ct=\xi=0$ and $x-ct=\xi=T_1$, or $(x_1(t),t), ~(x_2(t),t)$ where
\[
x_1(t)=ct,~~x_2(t)=T_1+ct.
\]
From the conditions \eqref{eq:conditionsforcat}, one can easily verify that condition in \eqref{eq:jumpcd1} is satisfied.
Hence, the patched function is a periodic  traveling peakon weak solution.
\end{proof}

One possible way to construct the arcs $\gamma^+$ and $\gamma^-$ in Proposition \ref{pro:patching1} is to cut curves from the same level set $H=h$. Below we show that the hyperbola $\phi^2-\frac{1}{3}v^2=c$ can be used to cut these two arcs from the level set $H=h$ so that the condition for endpoints \eqref{eq:conditionsforcat} is satisfied automatically. 
\begin{corollary}\label{cor:patching2}
Fix $k\in\mathbb{R}$ and $c\in\mathbb{R}$. Suppose $\gamma\subset \{(\phi, v): H=h\}\cap P_+$ is the graph of a function $v=f(\phi)$ that does not pass through any critical point. Assume the endpoints of $\gamma$, denoted as $A^+$ and $B^+$, are either on the hyperbola $\phi^2-\frac{1}{3}v^2=c$ or on the $v=0$ axis. Let $\tilde{\gamma}$ be the reflection of $\gamma$ about $v=0$ with corresponding endpoints $A^-, B^-$. Then the loop 
$\Gamma=\gamma\cup B^+B^-\cup\tilde{\gamma}\cup A^-A^+ $ gives a patched periodic traveling peakon weak solution of the mCH equation with wave speed $c$.
\end{corollary}

\begin{proof}
Assume $A^+=(\phi_1,v_1^+)$ and $B^+=(\phi_2,v_2^+)$. Due to the symmetric of $\gamma$ and $\tilde{\gamma}$, we have $A^-=(\phi_1,-v_1^+)$ and $B^-=(\phi_2,-v_2^+)$.
If $A^+,B^+$ are on the hyperbola $\phi^2-\frac{1}{3}v^2=c$, we know $A^-$ and $B^-$ are also on the hyperbola.  At this time, \eqref{eq:conditionsforcat} is satisfied and the result follows.

Consider that some endpoint, for example $A^+$, is on $v=0$. Then, $A^+=A^-$. Hence,  $\gamma\cup\tilde{\gamma}$ is actually a connected arc of the level set,  which forms a strong solution. The jump only happens at $B^+B^-$ and the jump conditions can be verified similarly as in Proposition \ref{pro:patching1}.
\end{proof}

Now,   for any traveling speed $c\in\mathbb{R}$ and any $k\in\mathbb{R}$, we prove the existence of the patched periodic traveling peakon weak solutions corresponding to Corollary \ref{cor:patching2}. And then give a figure to illustrate it (See Figure  \ref{fig:solutions}).
\begin{theorem}\label{thm:travelingweak}
 For any fixed traveling speed $c\in\mathbb{R}$ and any fixed $k\in\mathbb{R}$, there is a $h_0$ (depending on $k,c$ and $g$) such that for any $h> h_0$, there is an arc in the level set $H=h$  satisfying the conditions in Corollary \ref{cor:patching2}. Consequently, there are patched single-valued periodic traveling peakon weak solutions with speed $c$.
\end{theorem}
\begin{proof} 

We just consider the following function graph $v=f_1(\phi)$ in  $\Gamma_+$
\begin{gather*}
v=\sqrt{\phi^2-c+\sqrt{c^2-4(k\phi^2-g\phi-h)}}.
\end{gather*}
Consider \[
D=\{\phi\in\mathbb{R}: c^2-4(k\phi^2-g\phi-h)\ge 0\}.
\] 
When $k=0$, $D$ is an interval of $\mathbb{R}$. When $k<0$ and $h>\frac{-g^2-kc^2}{4k}$, we have $D=\mathbb{R}$. When $k>0$ and $h>\frac{-g^2-kc^2}{4k}$, $D$ is an interval of $\mathbb{R}$. Next, we assume $h$ big enough such that $D$ is an interval  of $\mathbb{R}$.

The domain of $v=f_1$ is given by \[
D_{f_1}=\{\phi\in D: \phi^2-c+\sqrt{c^2-4(k\phi^2-g\phi-h)}\ge 0 \}.
\]
For each $\tilde{\phi}$ in $\{\phi: \phi^2\le c\}$, we can find a $h_1(\tilde{\phi})$ such that 
$$\tilde{\phi}^2-c+\sqrt{c^2-4(k\tilde{\phi}^2-g\tilde{\phi}-h_1(\tilde{\phi}))}\ge 0.$$
Since $\{\phi: \phi^2\le c\}$ is a compact set, there exits $h_1>0$ such that for all $h>h_1$, $D_{f_1}=D$ is an interval  of $\mathbb{R}$.  

We now show that there exists $h_0\ge h_1$ so that whenever $h>h_0$, $v=f_1(\phi)$ intersects with $\phi^2-\frac{1}{3}v^2=c$ at two points $\phi_1, \phi_2$ and $[\phi_1, \phi_2]\subset D$.

Consider the following function
\[
f(\phi)=3(\phi^2-c)-(\phi^2-c+\sqrt{c^2-4(k\phi^2-g\phi-h)}).
\]

We claim that there is $h_2$ such that $\forall h\ge h_2$, $\exists M_1>0$, $[-M_1, M_1]\subset D$, $f(\pm M_1)>0$.
In the case $k<0$, when $h_0$ is large enough, $D=\mathbb{R}$, and the claim is clear since $f(\phi)$ grows like $2\phi^2$. In the case $k=0$, if $g=0$, $D=\mathbb{R}$. If $g>0$, $\phi\ge \phi_c=-\frac{c^2+4h}{4g}$. It is clear that $f(\phi_c)$ grows like $2h^2$. Hence, when $h$ is large enough, picking $M_1=|\phi_c|$ suffices. If $k>0$, the critical values are $\phi_c^{\pm}=(g\pm\sqrt{g^2+4h+c^2})/2k$. As $h$ increases, $D=[\phi_c^-, \phi_c^+]$ increases and $f(\phi_c^{\pm})$ increases to infinity. The claim is also true if we choose $M_1=\min(|\phi_c^-|, |\phi_c^+|)$.

Now, we compute
\[
f(0)=-2c-\sqrt{c^2+4h}<0.
\]
When $c>0$, we have $f(0)<0$. When $c<0$, $h> \frac{3c^2}{4}$ implies $f(0)<0$.  

Choose
\[
h_0=\max\Big\{h_1, h_2, \frac{3c^2}{4}\Big\}.
\]
Whenever $h>h_0$, $v=f_1(\phi)$ intersects with $\phi^2-\frac{1}{3}v^2=c$ at two points $\phi_1, \phi_2$ and $[\phi_1, \phi_2]\subset D$.

The proof then is complete.
\end{proof}

\begin{remark}
In the proof, we only constructed the cases when the two endpoints are on the hyperbola $\phi^2-\frac{1}{3}v^2=c$.  Actually, the other situation that one endpoint is on $v=0$ in Corollary \ref{cor:patching2} can also happen. See Figure \ref{fig:solutions} (a), (b)  for an example.  And by the patching technic, we can also obtain peakon weak solutions that are not periodic (see Remark \ref{rmk:solitonwave} and Section \ref{sec:peakonsolutions}).

\end{remark}

In Figure \ref{fig:solutions},  we plot two multi-valued traveling wave solutions and the corresponding patched periodic traveling peakon weak solutions for the cases  $k=1$, $c=2$ and $g=2$.

\begin{figure}[H]
\begin{center}
\includegraphics[width=0.8\textwidth]{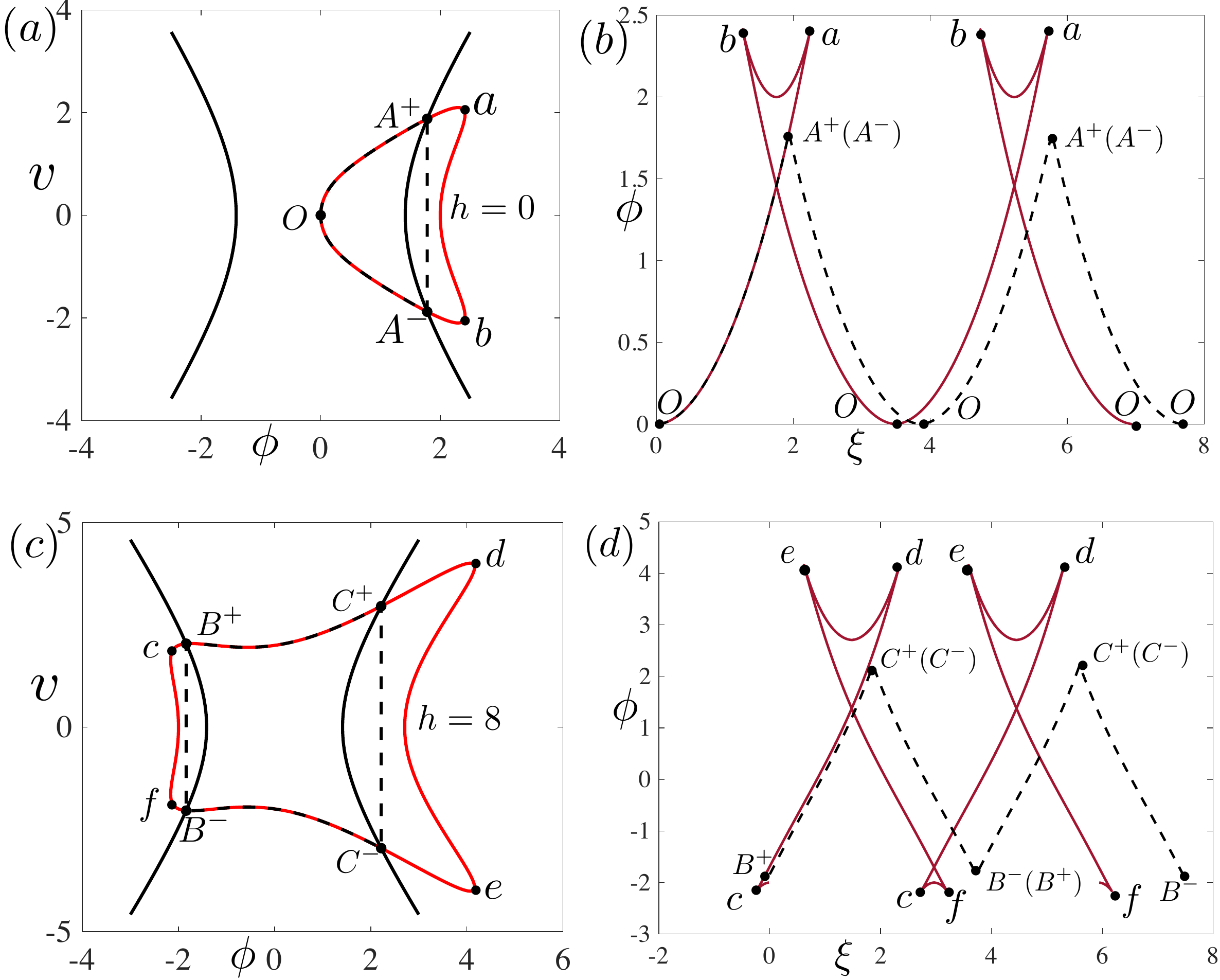}
\end{center}
\caption{In phase plane (a), the closed loop $\wideparen{OabO}$ (solid line) corresponds to multi-valued periodic  traveling wave solution $\wideparen{OabO}$ (solid line) in (b). This multi-valued solution was also given in \cite[Section 3]{Matsuno}, where it was called symmetric singular soliton. By the patching technic, we truncate the closed loop in the phase plane (a) at $A^+$ and $A^-$ and obtain a new closed loop  $\wideparen{OA^+A^-O}$ (dashed line) corresponding to a patched single-valued periodic traveling wave weak solution $\wideparen{OA^+A^-O}$ with a peakon $A^+(A^-)$ (dashed line) in (b). The points $A^+$ and $A^-$ are obtained by the interaction between the level set $H=0$ and the hyperbola $\phi^2-\frac{1}{3}v^2=c=2$ in phase plane. In (b), they glue together and form a peakon $A^+(A^-)$ where the jump condition \eqref{eq:jumpcd1} is satisfied (see Corollary \ref{cor:patching2} for more details).  The case for (c) and (d) is similar, where $B^+$, $B^-$ glue together and form a peakon at bottom while $C^+$, $C^-$ glue together and form a peakon at the top.} 
\label{fig:solutions}
\end{figure}

\begin{remark}\label{rmk:solitonwave}
Note that we can also patch solutions corresponding to different speeds. For example, consider the level set $H=0$ with the parameter $g=0$. Then, $(0,0)\in \{(\phi,v):H=0\}$. As in Figure \ref{fig:solutions} (b) (with $g=0$), assume $(\phi(\xi_0),v(\xi_0))=(0,0)$. Consider one period in $(\xi_0,\xi_0+T)$. Construct  $u$  as
\begin{align}\label{eq:patchdifspeed}
u(x,t)=\begin{cases}
0 ~\textrm{ for }~x-ct=\xi\notin [\xi_0,\xi_0+T],\\
\phi(\xi) ~\textrm{ for }~x-ct=\xi\in [\xi_0,\xi_0+T].
\end{cases}
\end{align}
Then, the jump condition \eqref{eq:continu} can be verified easily at the patching areas: $\xi_0$ and  $\xi_0+T$. Hence, $u$ is a weak solution to \eqref{mCH}. Based on this kind of peakon solutions,  we can construct multi-peakon solutions with different speeds.  For example, assume $u_1$ and $u_2$ are constructed as \eqref{eq:patchdifspeed} with different traveling speeds $c_1$, $c_2$ and periods $T_1$, $T_2$:
\begin{align*}
u_1(x,t)=\begin{cases}
0,~x-c_1t=\xi\notin [\xi_1,\xi_1+T_1],\\
\phi_1(\xi),~x-c_1t=\xi\in [\xi_1,\xi_1+T_1],
\end{cases}
u_2(x,t)=\begin{cases}
0,~x-c_2t=\xi\notin [\xi_2,\xi_2+T_2],\\
\phi_2(\xi),~x-c_2t=\xi\in [\xi_2,\xi_2+T_2].
\end{cases}
\end{align*}
For  $\xi_1+T_1< \xi_2$, we can construct $2$-peakon weak solution as
\begin{align*}
u(x,t)=\begin{cases}
u_1(x,t),~x-c_1t\in [\xi_1,\xi_1+T_1],\\
u_2(x,t),~x-c_2t\in [\xi_2,\xi_2+T_2],\\
0,~\textrm{otherwise}.
\end{cases}
\end{align*}
If $c_1<c_2$, the left peakon stays far behind the right one. If $c_1>c_2$, the left peakon catches up the right one in finite time. Whether the left peakon can go beyond the right one without changing shape or not, which is the properties of solitons, is left for future study.
\end{remark}

\section{A class of peakon weak solutions for $k\le0$}\label{sec:peakonsolutions}

In this section, we use the fundamental solution $G(x)$ (defined by \eqref{eq:fundamentalsolutionG1} and \eqref{eq:fundamental2}) of Helmholtz operator $1-\partial_{xx}$ to give some  peakon weak solutions to the mCH equation \eqref{mCH} (with $k\leq0$) in $\mathbb{T}_\ell$. And we show that some of these peakon weak solutions also can be obtained by the patching technic.

\subsection{Peakon weak solutions for $k=0$}
In this subsection, we consider the dispersionless mCH equation (Equation \eqref{mCH} with $k=0$).
When considering the whole space $\mathbb{R}$ (or $\ell=+\infty$), the dispersionless mCH equation has  peakon weak solutions of the form \eqref{solitarywave}. Here, we present a more general result about this kind of peakon weak solutions for all $0<\ell\leq +\infty$.
\begin{proposition}\label{pro:periodicpeakon}
Assume $0<\ell\leq+\infty$ and $p\neq 0$. Let
\begin{gather*}
c=\left\{
\begin{split}
&\frac{1}{4}p^2\Big[\coth^2(\frac{\ell}{2})-\frac{1}{3}\Big]~\textrm{ if }~\ell<+\infty;\\
&\frac{1}{6}p^2~\textrm{ if }~\ell=+\infty.
\end{split}
\right.
\end{gather*}
Then, $u(x, t)=pG(x-ct)$  is a one peakon weak solution to the dispersionless mCH equation \eqref{mCH} ($k=0$).
\end{proposition}
\begin{proof}
By the definition $G(x)$ \eqref{eq:fundamentalsolutionG1} and \eqref{eq:fundamental2},  we know
$$m(x,t)=u(x,t)-\partial_{xx}u(x,t)=0~\textrm{ for }~x\neq ct.$$
Hence, $u(x,t)$ is a classical solution to the dispersionless mCH equation in  $\mathbb{R}\times[0,+\infty)\setminus\{(x(t),t):x(t)=ct,~t\geq0\}$. By Theorem \ref{thm:jumpcondition}, we only have to prove $u(x,t)$ satisfies \eqref{eq:jumpcd1} or \eqref{eq:continu} along the line $\{(x(t),t):x(t)=ct,~t\geq0\}$. 
Here, we prove $u$ satisfies jump condition \eqref{eq:jumpcd1}. Notice that $u_x(x(t)-,t)=-u_x(x(t)+,t)=\frac{p}{2}$. We have
\begin{gather*}
u^2(x(t),t)-\frac{1}{3}u^2_x(x(t)-,t)-\frac{1}{3}u^2_x(x(t)+,t)-\frac{1}{3}u_x(x(t)-,t)u_x(x(t)+,t)\\
=\left\{\begin{split}
&\frac{1}{4}p^2\Big[\coth^2(\frac{\ell}{2})-\frac{1}{3}\Big]~\textrm{ if }~\ell<+\infty;\\
&\frac{1}{6}p^2~\textrm{ if }~\ell=+\infty.
\end{split}
\right.
\end{gather*}
This implies the first condition in \eqref{eq:jumpcd1}
\[
\frac{d}{dt}x(t)=c=u^2(x(t),t)-\frac{1}{3}u^2_x(x(t)-,t)-\frac{1}{3}u^2_x(x(t)+,t)-\frac{1}{3}u_x(x(t)-,t)u_x(x(t)+,t).
\]
The second condition in \eqref{eq:jumpcd1} can be easily checked. Therefore, $u(x,t)=pG(x-ct)$ is a peakon weak solution to the dispersionless mCH equation.

\end{proof}

\begin{remark}
The traveling speed is determined by the amplitude $p$ and the period $\ell$. When the period $\ell$ tends to infinity, the traveling speed $c$ tends to $\frac{1}{6}p^2$: 
$$\frac{1}{6}p^2=\lim_{\ell\to +\infty}\frac{1}{4}p^2\Big[\coth^2(\frac{\ell}{2})-\frac{1}{3}\Big].$$
\end{remark}

\begin{figure}[H]
\begin{center}
\includegraphics[width=0.8\textwidth]{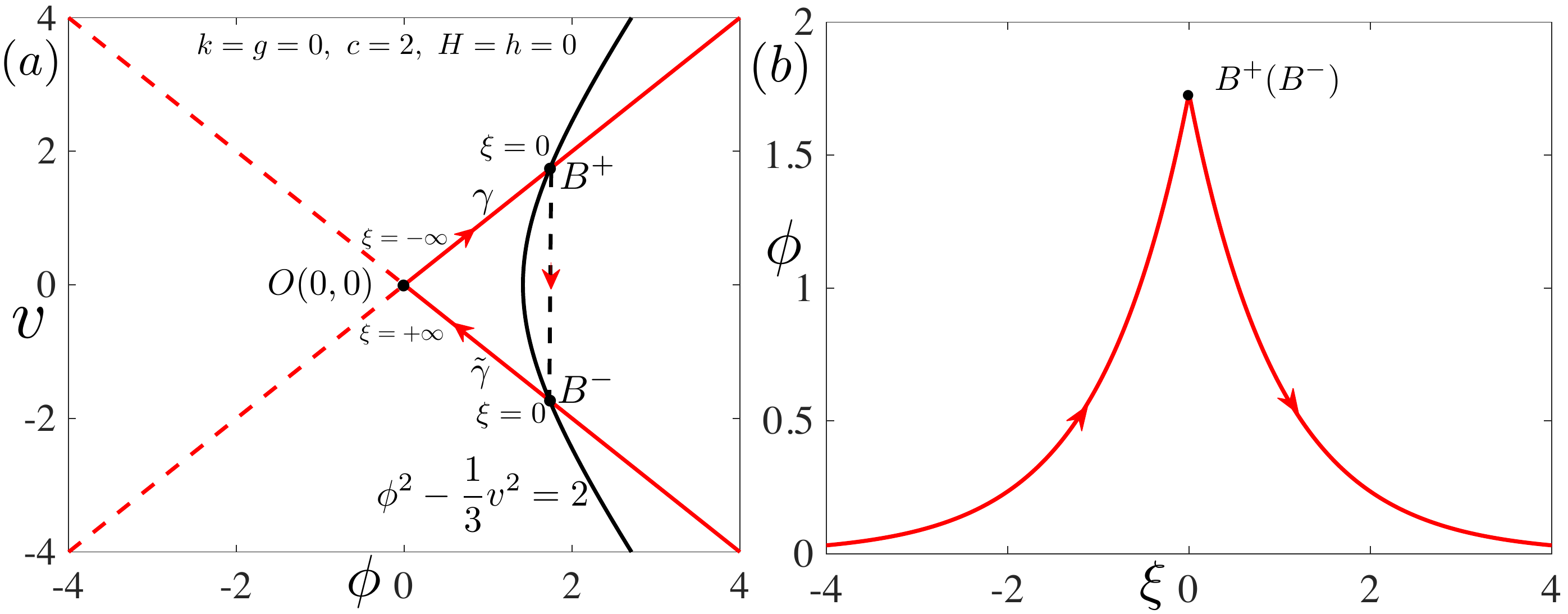}
\end{center}
\caption{Illustration of Remark \ref{remark:specialcase}: the patching technic for peakon weak solutions of the form \eqref{solitarywave}. When $g=k=0$ and traveling speed $c=2$, the hyperbola $\phi^2-\frac{1}{3}v^2=c$ ($\phi>0$) intersects the level set $H=0$ at $B^+$ and $B^-$ in phase plane (a). $(0,0)$ is a saddle point and $\gamma\cup B^+B^-\cup\tilde{\gamma}$ corresponds to the patched peakon weak solution $\phi(\xi)=u(x,t)=2\sqrt{3}G(x-2t)=\sqrt{3}e^{-|x-2t|}=\sqrt{3}e^{-|\xi|}$ in (b).} 
\label{fig:remark}
\end{figure}

\begin{remark}\label{remark:specialcase}
All the peakon weak solutions given by Proposition \ref{pro:periodicpeakon} can be obtained by the patching technic introduced in Subsection \ref{sec:patchedweaktravel}. Here, we only show the case for $\ell=+\infty.$
In this case, the solution is reduced to 
\[
u=pG(x-\frac{1}{6}p^2t)=\frac{1}{2}p e^{-|x-\frac{1}{6}p^2 t|}.
\] 

As shown in Figure \ref{fig:remark}, consider the level set of \eqref{eq:Hamiltonian} for $k=0,~g=0$  and $c>0$. At this time $(0,0)$ is a saddle point of \eqref{eq:system2}. The two lines $\phi=\pm v$ belong to the level set of $H=h=0$ and the time spent near $(\phi,v)=(0,0)$ is infinity. Consider the  two branches of the level set $H=0$ in the half plane $\phi>0$. These two branches correspond to two traveling wave solutions to the mCH equation:
$$\phi_1(\xi)=A_1e^{\xi}~\textrm{ and }~\phi_2(\xi)=A_2e^{-\xi}~\textrm{ for some constants }~A_i>0,~i=1,2.$$
(When the origin point for $\xi=0$ is set, we can determine $A_i$ for $i=1,2$.)
Now, assume the traveling speed $c=\frac{1}{6}p^2$ for $p\neq 0$. The two branches intersect $\phi=\sqrt{\frac{v^2}{3}+c}$ (the right curve of $\phi^2-\frac{v^2}{3}=c$) at $B^+=\Big(\sqrt{\frac{3}{2}c},\sqrt{\frac{3}{2}c}\Big)$ and $B^-=\Big(\sqrt{\frac{3}{2}c},-\sqrt{\frac{3}{2}c}\Big)$.
Assume $\gamma:=\Big\{(\phi,v):0<\phi<\sqrt{\frac{3}{2}c},~v=\phi\Big\}$ and $\tilde{\gamma}:=\Big\{(\phi,v):0<\phi<\sqrt{\frac{3}{2}c}, ~v=-\phi\Big\}$.
Consider the patched weak solution given by curve (clockwise)
$$\gamma\cup B^+B^-\cup\tilde{\gamma}.$$
Assume $\xi=0$ on the segment $B^+B^-$. Hence, we have
$$\phi_1(0)=\phi_2(0)=\sqrt{\frac{3}{2}c}=A_1=A_2.$$
This yields $A_i=\frac{|p|}{2}$ and the corresponding patched weak solution is 
$$u(x,t)=\phi(x-ct)=\frac{|p|}{2}e^{-|x-ct|}=|p|G(x-ct).$$
This is the one peakon weak solution given by Proposition \ref{pro:periodicpeakon}. This weak solution is not of finite period since it passes through the critical points $(0, 0)$. The time spent near $(0,0)$ is infinity. Similarly, the intersection between $\phi=\pm v$ ($\phi<0$) and the left curve of $\phi^2-\frac{v^2}{3}=c$ gives the peakon weak solution $u=-|p|G(x-ct).$

\end{remark}

Notice that a solution given by Proposition \ref{pro:periodicpeakon} satisfies $m(x,t)=0$ in the smooth region of $u$ and the superposition of such solutions $u$ gives a new solution in the smooth region. Hence, the superposition of $N$ such solutions with different traveling speeds  forms an $N$-peakon weak solutions only if one of the jump conditions (\eqref{eq:jumpcd1} or \eqref{eq:continu}) is satisfied along the trajectory of each peakon. 
Indeed, we have the following proposition.

\begin{proposition}\label{pro:NpeakonmCH}
Let  $p_1,\ldots,p_N\in\mathbb{R}$ be $N$ constants and $p_i\neq 0$, $i=1,\ldots,N$. Assume $x_{10}< x_{20}< \cdots< x_{N0}$. Consider $N$ trajectories $\{x_i(t)\}_{i=1}^N$ defined by
\begin{gather}\label{eq:NpeakonODE}
\frac{d}{dt}x_i(t)=u^2(x_i,t)-\frac{1}{3}\Big(u_x^2(x_i+,t)+u_x(x_i+,t)u_x(x_i-,t)+u^2_x(x_i-,t)\Big),
\end{gather}
subject to $x_i(0)=x_{i0}$, where
\begin{gather*}
u(x, t)=\sum_{i=1}^N p_iG(x-x_i(t))
\end{gather*}
Then, $u$ is an $N$-peakon weak solution of the mCH equation \eqref{mCH} before the first collision time $T$:
\begin{align}\label{eq:collisiontime}
T:=\sup\big\{T_0:x_1(t)<x_2(t)<\cdots<x_N(t)~\textrm{ for }~t\in[0,T_0)\big\}.
\end{align}

\end{proposition}
\begin{proof}
We only have to check the second condition in \eqref{eq:jumpcd1} for $t\in[0,T)$, where $T$ is defined by \eqref{eq:collisiontime}. By the assumption that $x_i(t)\neq x_j(t)$ for $i\neq j$ and $t<T$, we have
$$v_\ell^i(t)-v_r^i(t)=u_x(x_i(t)-,t)-u_x(x_i(t)+,t)=p_i.$$
Because $m=0$ in the smooth region of $u$, we have
$$A_\ell^i(t)=\lim_{x\rightarrow x_i(t)-}m(x,t)\Big(u^2(x,t)-u_x^2(x,t)-x'_i(t)\Big)=0$$
and similarly $A_r^i(t)=0.$
Hence, we have
$$\frac{d}{dt}(v_\ell^i(t)-v_r^i(t))=A_\ell^i(t)-A_r^i(t)=0,~~i=1,\ldots,N,$$
which proves the second condition of \eqref{eq:jumpcd1}.

\end{proof}

When $\ell<+\infty$ and $x_1(t)<x_2(t)<\cdots<x_N(t)$, equation\eqref{eq:NpeakonODE} can be simplified as 
\begin{multline}
\frac{d}{dt}x_i=\frac{1}{4}p_i^2\Big(\coth^2(\ell/2)-\frac{1}{3}\Big)+\coth(\ell/2)p_i\sum_{j\neq i}p_j G(x_i-x_j)\\
+\Big(\sum_{j\neq i}p_jG(x_i-x_j)\Big)^2-\Big(\sum_{j\neq i}p_j G_x(x_i-x_j)\Big)^2.
\end{multline}
In the case $\ell=\infty$ and $x_1(t)<x_2(t)<\cdots<x_N(t)$, Equation \eqref{eq:NpeakonODE} is simplified as 
\begin{equation}\label{eq:NpeakonmCH}
\frac{d}{dt}x_i=\frac{1}{6}p_i^2+\frac{1}{2}\sum_{j\neq i}p_ip_je^{-|x_i-x_j|}+\sum_{1\leq m<i<n\leq N}p_mp_ne^{-|x_m-x_n|}.
\end{equation} 
This ODE system was also obtained in \cite{GuiLiuOlverQu}.

\begin{remark}
Similar peakon weak soltuions can also be obtained for the CH equation \eqref{eq:generalCH} (with $k=0$) in $\mathbb{T}_\ell$. When $\ell=+\infty$, $u(x,t)=pG(x-ct)$ is a weak solution to the CH equation if the traveling speed is  $c=\frac{p}{2}$ \cite{Camassa}. When consider the periodic domain $\ell<+\infty$, one can prove that the traveling speed for such a peakon weak solution is given by $\frac{p}{2}\coth(\frac{\ell}{2})$ (see Appendix \ref{app:CHpeakon}). Letting $\ell\rightarrow +\infty$ in $\frac{p}{2}\coth(\frac{\ell}{2})$, the speed converges to $\frac{p}{2}$.

The periodic peakon weak solutions to the CH equation was also discussed in \cite{Constantin2}. But the traveling speed is given by $\frac{p}{2}$ instead of $\frac{p}{2}\coth(\frac{\ell}{2})$. In \cite{Lenells}, Lenells analyzed the trajectory stability of this periodic peakon weak soutions with the expression given by \cite{Constantin2}, without worrying about the speed of the traveling peakon. The reason that Lenells can sucessfully obtain the stability result is that he  used the method in \cite{stability}  which only considers the shape stability of peakon weak solutions and notice that different traveling speeds in $pG(x-ct)$ give a same function shape.
 \end{remark}

\subsection{Peakon weak solutions for $k< 0$.}\label{sec:knegative}
If the dispersive term $2ku_x$ is present for $k\in \mathbb{R}$, the CH equation \eqref{eq:generalCH} has peakon weak solutions of the form \cite{lwj04,Qian}:
\begin{align}\label{eq:peakonkleq}
u(x,t)=pe^{-|x-(p-k)t|}-k, ~~x\in\mathbb{R},~t\geq0.
\end{align}
Clearly, these peakon solutions do not vanish at infinity when $k\neq 0$($u \rightarrow-k$ as $|x|\to\infty$). 

For the mCH equation with $k<0$, we have the following similar result.
\begin{proposition}\label{pro:peakonknegative}
Assume $p\neq0$ and $k<0$. Let
\begin{align}\label{eq:peakonknegative}
u(x,t)=pG(x-ct)-\sqrt{-k}, ~m(x,t)=p\delta(x-ct)-\sqrt{-k},
\end{align}
where the traveling speed $c$ is given by
\begin{gather}\label{eq:speedforknegative}
c=\left\{
\begin{split}
&\Big[\frac{p}{2}\coth(\frac{\ell}{2})-\sqrt{-k}\Big]^2-\frac{p^2}{12}~\textrm{ if }~\ell<+\infty;\\
&\Big(\frac{p}{2}-\sqrt{-k}\Big)^2-\frac{p^2}{12}~\textrm{ if }~\ell=+\infty.
\end{split}
\right.
\end{gather}
Then, $u(x,t)$ is a weak solution to the mCH equation \eqref{mCH} in $\mathbb{T}_\ell.$
\end{proposition}
The proof is just a straightforward verification of the jump conditions proposed in Theorem \ref{thm:jumpcondition} and we omit it.
\begin{remark}\label{rmk:specialknegative}
We remark that peakon weak solutions \eqref{eq:peakonknegative} can also be obtained by the patching technic described in Subsection \ref{sec:patchedweaktravel}. Consider the case $\ell=+\infty$. Fix $k<0, ~k\neq-1$ and $p>0$. Let $c$ be defined by \eqref{eq:speedforknegative}. Set the integral constant $g=(c-k)\sqrt{-k}$ in \eqref{eq:Hamiltonian}. There are two branches of level set  $H=-k\Big(\frac{c}{2}-\frac{k}{4}\Big)$ intersects at  $(-\sqrt{-k},0)$ which is a saddle point of System \eqref{eq:system2}. With similar arguments as in Remark \ref{remark:specialcase}, one can show that these peakon weak solutions can be obtained by the patching technic. When $k=-1$, the critical point $(-1,0)$ is degenerate, but a patched peakon weak solution of form \eqref{eq:peakonknegative} can also be obtained.
\end{remark}

As shown by Proposition \ref{pro:periodicpeakon} and \ref{pro:peakonknegative}, the mCH equation \eqref{mCH} with $k\leq0$ has special peakon weak solutions constructed by the fundamental solution to the Helmholtz operator $1-\partial_{xx}$. And all these solutions can be constructed by the patching technic. Naturally, one may wonder whether it has such peakon weak solutions when $k>0$. Next, we give a proposition to show that there are no such peakon weak solutions to \eqref{mCH} when $k>0$.
\begin{proposition}\label{pro:nopeakonkpositive}
If $k>0$, there is no weak solution to \eqref{mCH} of the form $u(x,t)=AG(x-ct)-B$ for any constants $c$, $A\neq0$ and $B$.
\end{proposition}
\begin{proof}
Suppose there is a weak solution to \eqref{mCH} of the form $u(x,t)=AG(x-ct)-B$. Then, $m=u-u_{xx}=-B=:M$ is a constant when $x\neq ct$. Consider the smooth region of $u$ where it is a traveling wave solution.  Let $\phi(\xi):=u(x,t)$ and $\xi=x-ct$.  By System \eqref{eq:system},  we have
\[
m=\phi-\phi''=\frac{2k\phi-g}{c-\phi^2+v^2}=M.
\]
Since $\phi$ is not a constant, we find 
\[
\phi^2-v^2-c=\frac{1}{M}(g-2k\phi).
\]
The solution must correspond to a level set of $H$:
\begin{align}\label{eq:mconstant}
H&=\frac{1}{4}(\phi^2-v^2)^2+\frac{1}{2}c(v^2-\phi^2)+k\phi^2-g\phi \nonumber\\
&=\frac{1}{4}\Big(c+\frac{1}{M}(g-2k\phi)\Big)^2-\frac{1}{2}c\Big(c+\frac{1}{M}(g-2k\phi)\Big)+k\phi^2-g\phi\nonumber\\
&=\frac{k(k+M^2)}{M^2}\phi^2-\frac{g(k+M^2)}{M^2}\phi+\frac{1}{4}\Big(c+\frac{g}{M}\Big)^2-\frac{1}{2}c\Big(c+\frac{g}{M}\Big)=h.
\end{align}
Since $\phi$ is not a constant, the coefficients of $\phi$ must be zero 
on the right hand side but this is impossible if $k>0$.
\end{proof}
\begin{remark}
When $k<0$, \eqref{eq:mconstant} also implies that there exists traveling wave solutions  that makes $m$ is a constant in smooth regions only if $M=\sqrt{-k}$.
\end{remark}

\section{Planar curve flows and the mCH equation}\label{sec:curve}
In this section, we discuss the planar curve flow $\mathscr{C}(t)$ and see how a class of curve flows are related to the mCH equation. Suppose the velocity for a point $p(t;\xi)$ on curve $\mathscr{C}(t)$ is given by $v(t;\xi)$, and thus
\[
\frac{dp(t;\xi)}{dt}=v(t;\xi),
\]
where $\xi\in [0, \ell]$ is a label for the point, called Lagrangian coordinate for the curve. The Lagrangian coordinate allows us to parametrize $\mathscr{C}(t)$ as $r(\xi, t):=p(t;\xi)$. Since the velocity $v(t;\xi)$ can be decomposed along the tangential and normal directions, we then can write the equation for the planar curve flow as
\begin{gather}\label{eq:curveflow}
r_t=a\tau+b n, 
\end{gather}
where 
 \[
\tau=\tau(\xi,t)=(\tau_1,\tau_2):=r_{\xi}/|r_{\xi}|\] 
is the unit tangent vector and \[
n=n(\xi,t):=\tau^\perp(\xi,t)=(-\tau_2, \tau_1)
\] is the unit normal vector.

Before we discuss the our planar curve flows, let us make the following remarks: (i). If $a$ and $b$ are Euclidean differential invariants (i.e. they only depend on curvature and its derivatives about the arc-length), the curve flows are intrinsically the same (i.e. if $r$ is the flow for $r_0$, then $Pr+b$ is the flow for initial curve $Pr_0+b$ where $P$ is a rotation and $b$ is a translation). (ii). Geometric properties like total arclength, curvature and the derivatives of curvature only depend on the normal component $b$. This means if one cares the geometric curve flow only, one can set $a$ to be anything (see \cite{olver08}). (iii). $a$ affects the stretching of the curve and elastic properties.

 We will study curve motions that preserve arc-length and total signed area in detail. For relative discussions, see \cite{olver08}.

\subsection{Arc-length preserving motions}
The following results give criteria for preserving arc-length and total signed area of closed curve flows given by \eqref{eq:curveflow}.
\begin{lemma}\label{lmm:arcareapre}
Suppose a smooth curve flow $r(\xi,t)$ is a solution to \eqref{eq:curveflow} with Euclidean differential invariants $a$ and $b$.  $\xi\in[0,\ell]$ is the Lagrangian coordinate and $s=s(\xi,t)$ is the arc-length parameter at time $t$.  Let $\kappa(s,t)$ be the curvature  of the curve $r(\xi,t)$. 
Then, the following statements hold.

(i). The curve is non-stretching (i.e. $\partial_t|r_\xi(\xi,t)|\equiv0$,  the length between any two points is unchanged or arc-length is preserved) if and only if 
\begin{gather}\label{eq:arclengthpre}
a_s=b\kappa
\end{gather}
holds for any $t>0.$ 

(ii). Suppose the curve is a closed curve, or in other words, $r(\xi, t)$ is a periodic function with period $\ell$ and $\xi\in [0, \ell]$.  Then, we have
\begin{gather}\label{eq:areapre}
\frac{d}{dt}A(t)=-\oint_{\mathscr{C}(t)} b ds,
\end{gather}
where 
\begin{align}\label{eq:totalsignedarea}
A(t)=\frac{1}{2}\oint_{\mathscr{C}(t)} r\cdot(-n ds)
\end{align}
 is the total signed area enclosed by the curve $r(\xi,t)$.
In particular, if $\oint_{\mathscr{C}(t)} b ds=0$, the curve flow preserves the area enclosed.
\end{lemma}

\begin{proof}
(i).
Denote
\begin{gather*}
\chi(\xi,t)=|r_{\xi}(\xi,t)|.
\end{gather*} 
Because 
\[
\frac{1}{2}\partial_t\chi^2=r_{\xi}\cdot r_{\xi t}=a_{\xi}\chi-\chi^2 b\kappa,
\]
we see that $\chi$ is a constant if and only if $a_{\xi}=b\kappa \chi$ or equivalently \eqref{eq:arclengthpre} holds since $ds=\chi d\xi$.  

(ii). By direct computation, we find 
\begin{gather}\label{eq:taut}
\tau_t=\partial_t\left(\frac{r_{\xi}}{|r_{\xi}|}\right)=\frac{1}{\chi}(a\tau_{\xi}+b_{\xi}n)=\left(a\kappa+\frac{b_{\xi}}{\chi}\right)n.
\end{gather}

Further, since $n_t\perp n$, we can write $n_t=h\tau$ for some constant $h$. By the fact that $\tau\cdot n=0$, we have $h=-n\cdot\tau_t$ and \[
n_t=-\left(a\kappa+\frac{b_{\xi}}{\chi}\right)\tau.
\] 
We now rewrite the formula for the total signed area as
\[
A(t)=\frac{1}{2}\oint_{\mathscr{C}(t)} r\cdot(-n ds)=-\frac{1}{2}\int_0^\ell r\cdot n \chi d\xi.
\]
Hence, 
\[
\frac{d}{dt}A=-\frac{1}{2}\int_0^\ell b|r_{\xi}|+r\cdot\tau(-(a\kappa|r_{\xi}|+b_{\xi}))+r\cdot n(a_{\xi}-|r_{\xi}|b\kappa)d\xi
\]
Integrating by parts for the terms about $a_{\xi}$ and $b_{\xi}$, we have $A'=-\int_0^\ell b|r_{\xi}|d\xi=-\oint_{\mathscr{C}(t)} b ds$.

\end{proof}
Under the non-stretching assumption \eqref{eq:arclengthpre} for curve flow $r(\xi,t)$, we derive an evolution equation for curvature $\kappa$. When $s(\xi,0)=\xi$, by Lemma \ref{lmm:arcareapre} we have $\chi(\xi,t)\equiv1$ and $s(\xi,t)\equiv\xi$ for any $t\geq0$. Consequently, \eqref{eq:taut} can be rewritten as
\[
\tau_t=(b_s+\kappa a){n}.
\]
Using the above equation, we can obtain
\begin{align*}
(\kappa n)_t=\tau_{st}=\tau_{ts}&=(b_s+\kappa a)_s{n}+(b_s+\kappa a){n_s}\\
&=(b_{ss}+(\kappa a)_s)n-\kappa(b_s+\kappa a)\tau.
\end{align*}
Hence, it follows that
\begin{gather}\label{eq:curvature}
\kappa_t=n\cdot (\kappa n)_t=b_{ss}+(\kappa a)_s,
\end{gather}
where $n_t\cdot n=0$ was used.

We give two examples of closed curve flows connecting with two integrable systems which preserve arc-length and total signed area. These two examples can be found in \cite{GuiLiuOlverQu}.

\begin{enumerate}
  \item Take $a=\frac{1}{2}\kappa^2$, $b=\frac{1}{\kappa}a_s=\kappa_s$ in \eqref{eq:curveflow}. The arc-length parameter is then preserved for any $t$.  In the case that the curve is closed, $\oint_{\mathscr{C}(t)} b ds=0$ and the total signed area enclosed by the curve is preserved. 
By \eqref{eq:curvature},  the curvature satisfies 
\begin{gather}\label{eq:mKdV}
\kappa_t=\kappa_{sss}+\frac{3}{2}\kappa^2\kappa_s.
\end{gather}
  This equation is the modified Korteweg-de Vries(mKdV) equation.
\item In this example, $a$ and $b$  in \eqref{eq:curveflow} are determined by $\kappa$ globally by first solving $u-u_{ss}=\kappa$ and then taking $a=-(u^2-u_s^2)-2$, $b=a_s/\kappa=-2u_s$. In the case that the curve is closed, we have
$\oint_{\mathscr{C}(t)} b ds=0$
which means the total signed area is also preserved. By \eqref{eq:curvature}, $u$ satisfies the following mCH equation with a dispersion term $2u_s$ 
\begin{gather}\label{eq:mCHcurve}
\kappa_t+[(u^2-u_s^2)\kappa]_s+2u_s=0, \quad \kappa=(1-\partial_{ss})u.
\end{gather}
Hence, the curvature is the function $m$  introduced in \eqref{mCH}. This example connects a particular curve flows with the mCH equation we are studying in this paper.
\end{enumerate}

\begin{remark}
It is well-known that a planar curve is determined uniquely by its curvature function up to a rigid body motion. In the recent work \cite{Mumford}, Bruveris, Michor and Mumford showed that the space of closed plane curves equipped with some Sobolev-type metric is geodesically complete.
\end{remark}

\subsection{Hamiltonian structure for the arc-length preserving motions}

In this section, we consider closed curve flows given by a Hamiltonian system
\begin{gather}\label{eq:variationHamil}
r_t=P\frac{\delta E}{\delta r}
\end{gather}
where $E$ is some functional of the curve $r$ and $P$ is an anti-Hermitian operator so that $E$ is conserved under this motion. One example is 
\begin{align}\label{eq:J}
P=\lambda(\xi)J, ~~J=\left(
\begin{array}{cc}
0 &1\\
-1 &0\\
\end{array}
\right)
\end{align}
and $\lambda$ is a scalar. We aim to find Hamiltonian systems of the form \eqref{eq:variationHamil} for the arc-length preserving flows \eqref{eq:curveflow}.

Motivated by the elastic energies in material science, we consider energy functionals depending both on geometric quantities (curvature and its derivatives) and stretching. To define the stretching, we assume the curve has a relaxed configuration so that it contains no elastic energy under this configuration (one may consider a planar elastic spring for the motivation). Choosing $\xi$ as the arc-length parameter for this relaxed configuration, then the elastic stretching is defined as 
\begin{gather}\label{eq:stretching}
\chi(\xi,t)=|r_{\xi}(\xi,t)|.
\end{gather}

We have the following proposition.
\begin{proposition}\label{pro:variation1}
Suppose a smooth curve flow $r(\xi,t)$ is a solution to \eqref{eq:curveflow} with Euclidean differential invariants $a$ and $b$.  $\xi\in[0,\ell]$ is the Lagrangian coordinate chosen as arc-length parameter for the relaxed configuration and $s=s(\xi,t)$ is the arc-length parameter at time $t$.  Let $\kappa(s,t)$ be the curvature  of the curve $r(\xi,t)$.  Denote $\chi(\xi, t)$ the elastic stretching of $r(\xi,t)$ at time $t$ defined in \eqref{eq:stretching}.

 If $s(\xi,0)=\xi$ and $a_s=b\kappa$, then the dynamics \eqref{eq:curveflow} for this curve flow can be rewritten in the following Hamiltonian structure
\begin{gather}\label{eq:dynamic}
r_t=-\frac{a^2}{\kappa}J\frac{\delta E}{\delta r},
\end{gather}
where $J$ and $E$ are given by
\begin{align}\label{eq:Energyfunctional}
J=\left(
\begin{array}{cc}
0 &1\\
-1 &0\\
\end{array}
\right),~~E:=\int_0^{\ell} \frac{1}{a}(\chi-1)d\xi.
\end{align}
Consequently, the curve flows given by \eqref{eq:dynamic} move in the manifold $\{r: E=0\}$.

In particular, when $a=\frac{1}{2}\kappa^2$ in \eqref{eq:Energyfunctional}, \eqref{eq:dynamic} gives the curve flows described by the mKdV equation \eqref{eq:mKdV}.
When $a=-(u^2-u_s^2)-2$ in \eqref{eq:Energyfunctional} where $u$ satisfies $u-u_{ss}=\kappa$, \eqref{eq:dynamic} gives the curve flows described by the mCH equation \eqref{eq:mCHcurve}.

\end{proposition}

\begin{proof}

By Lemma \ref{lmm:arcareapre}, if $a_s=b\kappa$, then the curve is non-stretching and the arc-length is preserved. Hence, when $s(\xi,0)=\xi$, we have $s(\xi,t)\equiv\xi$ for any time $t\geq0$. 

Taking variation of the curve subject to $\delta r(0)=\delta r(\ell)$ gives 
\begin{gather}\label{eq:variationE}
\delta E=-\int_0^{\ell} \frac{\delta a}{a^2}(\chi-1)d\xi
-\int_0^\ell\frac{\partial}{\partial\xi}\left(\frac{1}{a}\frac{r_{\xi}}{|r_{\xi}|}\right)
\cdot \delta r d\xi.
\end{gather}
Along the arc-length preserving curve flow $r(\xi,t)$, we have $\chi\equiv1$ and the first term in  \eqref{eq:variationE} is zero. Hence, we have
\begin{gather}
\frac{\delta E}{\delta r}=-\frac{\partial}{\partial s}\left(\frac{1}{a}\tau\right)=\frac{a_s}{a^2}\tau-\frac{\kappa}{a}n.
\end{gather}
 As a consequence, we have
\begin{gather*}
r_t=a\tau+bn=a\tau+\frac{a_s}{\kappa}n=-\frac{a^2}{\kappa}J\frac{\delta E}{\delta r}.
\end{gather*}

The rest statements can be obtained easily and we omit the proofs.
\end{proof}

\subsection{Loops with cusps}

We now study the closed curve flows corresponding to the patched traveling peakon weak solutions to the mCH equation \eqref{eq:mCHcurve}, which are constructed in Theorem \ref{thm:travelingweak}. Recall that such a patched peakon weak solution $u(s,t)=\phi(\xi)$ ($\xi=s-ct$) is piecewise smooth with jumps in first order derivative. The curvature $\kappa=u-u_{ss}=\phi-\phi''$ ($\phi''$ is understood in distribution sense) of the corresponding curve contains Dirac delta mass. If the curve is closed, then we have a loop with cusps due to the Dirac delta mass in curvature. 

The curvature $\kappa(s)$ of a closed curve must also be a periodic function and the period $T$ is a divisor of  the curve  length $\ell$. Moreover, it must satisfy the Gauss-Bonnet formula:
\begin{gather}\label{eq:GBformula}
\int_0^\ell\kappa(s)ds=2n\pi, ~ n\in\mathbb{Z}.
\end{gather}
For a simple  closed loop, we have $n=\pm 1$. If we allow the curve to turn over, then $n$ can be other integers. 

However, not every periodic solution $\kappa$ of Equation \eqref{eq:mCHcurve} corresponds to a closed curve flow. Next, we study what kind of $\kappa$ gives a closed curve. For convenience, we now introduce the angle function $\theta(s, t)$ which is the integral of $\kappa(s, t)$ about $s$.  Note that $\kappa$ is the distributional derivative of a piecewise smooth function $f$. We omit the $t$ dependence for the convenience of discussion. The angle function is then defined as
\begin{gather}\label{eq:int}
\theta(s):=\int_{(0, s]}\kappa(z) d z=f(s+)-f(0+).
\end{gather}
We now present a criterion on $\theta$ that guarantees that the curve is closed.  With a modification of results in \cite{agm08}, we have
\begin{lemma}\label{lmm:closedcurve}
Suppose $\theta(s)$ is a piecewise smooth function that is cadlag (right-continuous with left limits) and there exists $T>0$ such that $\theta(s+T)=\theta(s)+\theta(T)$, $\forall s\ge 0$.

If $\theta(T)\in 2\pi(\mathbb{Q}\setminus\mathbb{Z})$ so that 
\[
\theta(T)=\frac{2m\pi}{n}, m,n\in\mathbb{Z}, n>0, ~gcd(m,n)=1,
\] 
then \[
\gamma(s)=\int_0^{s}\exp(i\theta(\tau))d\tau
\]
represents a closed curve in the plane with $s$ to be the arc-length parameter. The total length of the curve is $nT$.
If $\theta(T)\in 2\pi\mathbb{Z}$, there are examples that the curves are not closed.
\end{lemma}
\begin{proof}
We show that \[
\gamma(nT+s)-\gamma(s)=0, \forall s.
\]
This is true because
\begin{multline*}
\int_{s}^{s+nT}\exp(i\theta(\tau))d\tau
=\sum_{j=0}^{n-1}\int_{s+jT}^{s+(j+1)T}\exp(i\theta(\tau))d\tau
=\sum_{j=0}^{n-1}\int_s^{s+T}\exp(i\theta(\tau+jT))d\tau \\
=\sum_{j=0}^{n-1}\exp(ij\theta(T))\int_s^{s+T}\exp(i\theta(\tau))d\tau
=\int_s^{s+T}\exp(i\theta(\tau))d\tau \frac{1-\exp(im2\pi)}{1-\exp(im2\pi/n)}=0.
\end{multline*}
Then, the total length of $\gamma$ must be $n_1T$ so that $n_1|n$ for some integer $n_1$. 

By the definition of $\gamma(s)$, $\theta$ is the incline angle of the tangent vector and hence $n_1\theta(T)$ is a multiple of $2\pi$ since the curve is closed, or $n|mn_1$. Since $gcd(n, m)=1$, we must have $n_1=n$ and  the total length of $\gamma$ is $nT$.

In the case $\theta(T)\in2\pi \mathbb{Z}$, one can refer to \cite{agm08} for examples where the curves are not closed.
\end{proof}

Now, we show that some patched periodic traveling peakon weak solutions of \eqref{eq:mCHcurve} may correspond to loops with cusps. Fix $k=1$,  $c>0$ and $ g>0$. By Theorem \ref{thm:travelingweak}, there exists $h_0=h_0(c,g)>0$  so that for all $h>h_0$, 
\begin{gather}\label{eq:gammah}
\Gamma_h:=\Big\{(\phi, v):v^2=\phi^2-c+\sqrt{c^2-4(\phi^2-g\phi-h)},~\phi_1\leq \phi\leq \phi_2,~v>0\Big\},
\end{gather}
gives a patched periodic traveling peakon solution $u_h(s,t)=\phi_h(\xi)$ ($\xi=s-ct$ and $\xi\in [0, T_h]$, $T_h$ is a periodic), where $\phi_1, \phi_2$ are determined by
\begin{align}\label{eq:intersection}
\begin{cases}
v^2=\phi^2-c+\sqrt{c^2-4(\phi^2-g\phi-h)},~ v>0,\\
\phi^2-\frac{1}{3}v^2=c.
\end{cases}
\end{align}

\begin{proposition}\label{pro:looptheta}
Assume $k=1$ and fix $c>0$ and $ g>0$. Let $u_h(s,t)=\phi_h(\xi): \xi=s-ct\in [0, T_h]$ be the patched periodic weak solution corresponding to $\Gamma_h$ defined by \eqref{eq:gammah} for $h\in(h_0, +\infty)$, where $T_h$ is a period of the patched periodic peakon solution and $h_0$ is the constant in Theorem \ref{thm:travelingweak}. Let $(\phi_1, v_1^+)$ and $(\phi_2, v_2^+)$ be the two endpoints of $\Gamma_h$ with $\phi_1<\phi_2$. The following statements hold:

(i). Let $\theta_h$ be the angle function corresponding to $\kappa_h(\xi):=\phi_h(\xi)-\phi_h(\xi)''$. Then, 
\begin{gather}\label{eq:thetah}
\theta_h(T_h)=\displaystyle 2\left(\int_{\Gamma_h}\frac{2\phi-g}{c-\phi^2+v^2}
\frac{1}{v}d\phi+v_2^+-v_1^+\right)=
\int_0^{T_h} \phi_h(\xi)d\xi=\displaystyle 2\int_{\Gamma_h} \frac{\phi }{v}d\phi.
\end{gather}

(ii).
$\theta_h(T_h)$ is continuous in $h$. If $\theta_h(T_h)$ is not a constant for $h\in (h_0, \infty)$, then there exist infinitely many solutions $\phi_h$ such that 
\[
\theta_h(T_h)\in 2\pi(\mathbb{Q}\setminus\mathbb{Z}).
\]
Letting $\theta_h(T_h)=2\pi\frac{n}{N}$ with $n, N\in\mathbb{Z}$, $N>0, gcd(n,N)=1$, then the solution gives an arc-length preserving loop with cusps. The total length is $L_h=NT_h$ and the curvature $\kappa_h$ satisfies 
\[
\int_{(0,L_h]}\kappa_h(\xi) d \xi=2n\pi.
\]

\end{proposition}

\begin{proof}
Since $\phi_h''$ is the distributional derivative of $\phi_h'$ which is periodic and piecewise continuous,  by Equation \eqref{eq:int}, we find that 
\[
\int_{(0, T_h]}\phi_h'' d\xi=0.
\]
Due to $d\xi=\frac{1}{v}d\phi$, we have
\[
\theta_h(T_h)=\int_0^{T_h}\kappa_h(\xi) d\xi=\int_0^{T_h} \phi_h d\xi=2\int_{\Gamma_h}\frac{\phi}{v}d\phi.
\]
By this expression, the continuity of $\theta_h(T_h)$ on $h$ is clear.

For the other expression, we just consider 
\begin{gather}
\int_{(0,T_h]}\kappa d\xi=2\bigg(\int_{\Gamma_h}\kappa \frac{1}{v}d\phi+\int_{\Gamma_h}\phi'' d\xi\bigg)
=2\int_{\Gamma_h}\kappa \frac{1}{v}d\phi+2(v_2^+-v_1^+).
\end{gather}
On $\Gamma_h$, by System \eqref{eq:system} we have
\[
\phi_h''=\phi_h+\frac{g-2\phi_h}{c-\phi_h^2+(\phi_h')^2} \Rightarrow \kappa_h(\xi)=\frac{2\phi-g}{c-\phi^2+v^2}\Big|_{(\phi, v)=(\phi_h(\xi), \phi_h'(\xi))}.
\]

Since rational numbers are dense in any interval of $\mathbb{R}$, as long as $\theta_h(T_h)$ is not constant in $h$, there are infinitely many $h$ such that $\theta_h(T_h)\in 2\pi(\mathbb{Q}\setminus\mathbb{Z})$. According to Lemma \ref{lmm:closedcurve}, the claim follows.
\end{proof}

Next, we give an example to illustrate Proposition \ref{pro:looptheta}. 

Fix $g=2$ and $c=2$. 
Consider the branch $\Gamma_+$ (defined by \eqref{eq:Gammapm}) of $H=h$ in the phase plane that $v>0$.
Let $f_{0h}(\phi):=4-4(\phi^2-2\phi-h)$ and $f_{1h}(\phi):=\phi^2-2+\sqrt{f_{0h}(\phi)}.$
Assume 
$$D_{0h}:=\{\phi\in\mathbb{R}: f_{0h}(\phi)\ge 0\}~\textrm{ and }~D_{1h}:=\{\phi\in\mathbb{R}: f_{1h}(\phi)\ge 0\}.$$
Choose $h_0=1+2\sqrt{2}$ and  for any $h>h_0$ we have
$$f_{0h}(\phi)\geq f_{0h}(-\sqrt{2})\geq 0~\textrm{ for any }~\phi\in[-\sqrt{2},\sqrt{2}].$$
Moreover, we have
$$f_{1h}(\phi)\geq0~\textrm{ for }~\phi\in[-\sqrt{2},\sqrt{2}], ~\textrm{ and }~[-\sqrt{2},\sqrt{2}]\subset D_{1h}.$$
Consider the intersections between $\Gamma_+$ ($v>0$) and $\phi^2-\frac{1}{3}v^2=2$. For $h\in(1+\sqrt{2},+\infty)$, there are two points $\phi_{1h},\phi_{2h}\in\mathbb{R}$ satisfy
\[
2(\phi^2-2)=\sqrt{4-4(\phi^2-2\phi-h)}.
\]
Hence, from \eqref{eq:thetah} we have
\[
\theta_h(T_h)=2\int_{\phi_{1h}}^{\phi_{2h}}\frac{\phi}{\sqrt{f_{1h}(\phi)}}d\phi=2\int_{\phi_{1h}}^{\phi_{2h}}\frac{\phi}{\sqrt{\phi^2-2+\sqrt{4-4(\phi^2-2\phi-h)}}}d\phi.
\]
Therefore, $\phi_h$ in Proposition \ref{pro:looptheta} is defined on $(1+2\sqrt{2},\infty)$.
\begin{figure}[H]
\begin{center}
\includegraphics[width=0.9\textwidth]{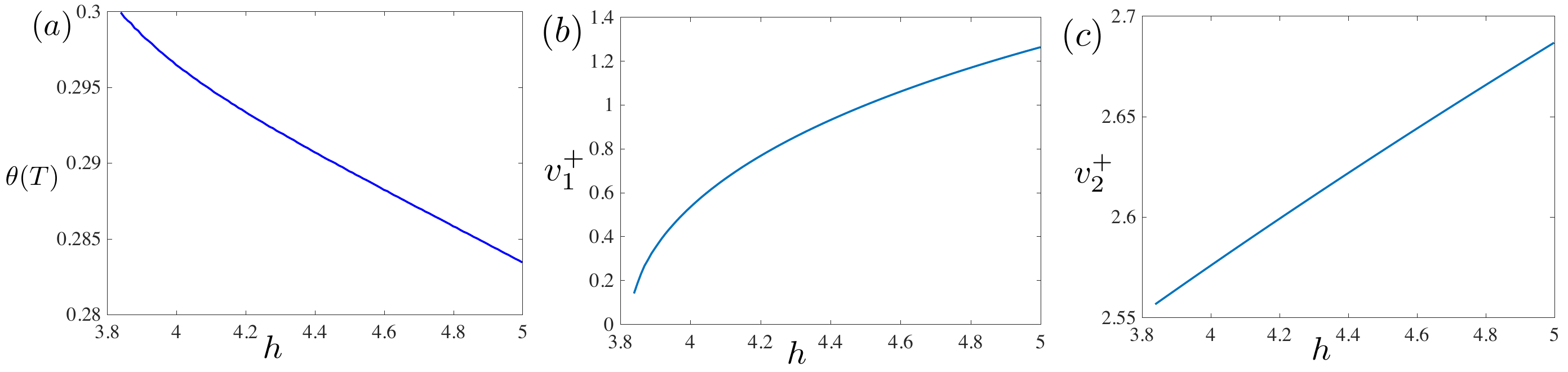}
\end{center}
\caption{(a): $\theta_h(T_h)$ versus $h$. Figure (b) and (c) shows the jump of $\theta_h(T_h)$ at the end of $\Gamma_h$.  }
\label{fig:thetaT}
\end{figure}

 The figure $\theta_h(T_h)$ versus $h$ is shown in Figure \ref{fig:thetaT}. Clearly, $\theta_h(T_h)$ is not a constant and therefore there are $h$  such that $\theta_h(T_h)\in 2\pi(\mathbb{Q}\setminus\mathbb{Z})$. Such a solution then gives a closed loop by Lemma \ref{lmm:closedcurve}.  
 
 However, it is not easy to tell which $h$ satisfies $\theta_h(T_h)\in 2\pi(\mathbb{Q}\setminus\mathbb{Z})$. To illustrate how the cusps happen, let us consider a function $\phi$, which is not necessarily a solution to the mCH equation: 
\[
\phi(\xi)=
\begin{cases}
\phi_0(\xi), & \xi\in[0,1],\\
\phi_0(2-\xi),& \xi\in (1, 2],
\end{cases}
\]
where $\phi_0(\xi)=\frac{\pi}{3}+\xi^2+\xi-\frac{5}{6}~\textrm{ for }~ \xi\in[0,1]$.
Then,  $\kappa=\phi-\phi''=\frac{\pi}{3}+\xi+\xi^2-\frac{17}{6}, \xi \in (0, 1)$. As a result, the angle function $\theta(s)$ is given by
\begin{gather}\label{eq:example}
\theta(s)=
\begin{cases}
\displaystyle{\frac{\pi}{3}s+\frac{1}{2}s^2+\frac{1}{3}s^3-\frac{17}{6}s,~ s\in [0,1),}\\
\displaystyle{\frac{\pi}{3}s-\frac{1}{3}(2-s)^3-\frac{1}{2}(2-s)^2-\frac{17s}{6}+\frac{23}{3},~ s\in [1,2).}
\end{cases}
\end{gather}
Then, $\theta(2)=\theta(2-)-2=\frac{2\pi}{3}$.

In Figure \ref{fig:curve}, we plot the $\theta(s)$ function and the corresponding curve. The curve is computed by evaluating the integral $\gamma(s)=\int_0^s\exp(i\theta(\tau))d\tau$ directly. 

\begin{figure}[H]
\begin{center}
\includegraphics[width=0.8\textwidth]{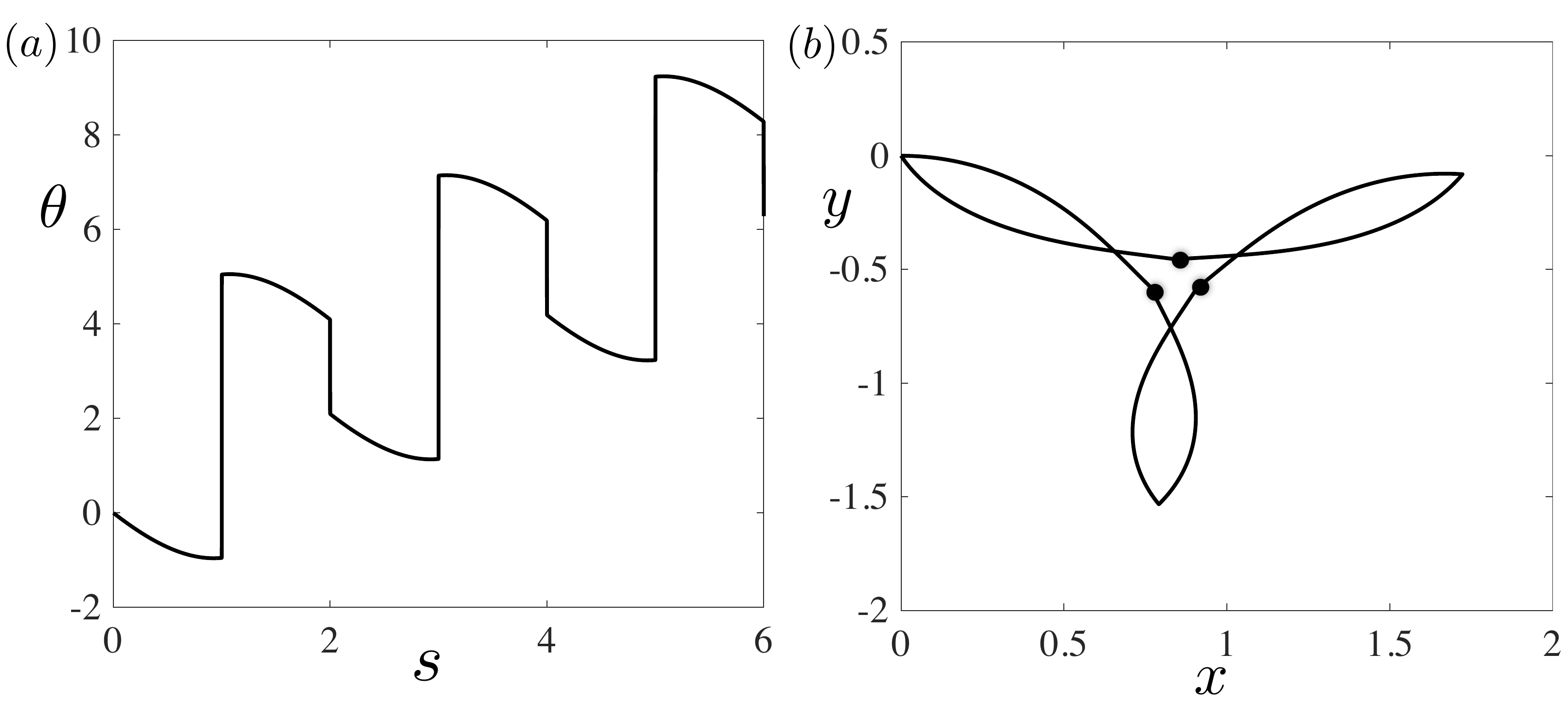}
\end{center}
\caption{Illustration of loops with cusps with Equation \eqref{eq:example}. Figure (a) shows the graphs of $\theta(s)$ function, and Figure (b) shows the corresponding planar curves. The dots show the singular points at $s=1,3,5$, while the other cusps correspond to $s=0,2,4$. }
\label{fig:curve}
\end{figure}

\appendix

\section{Proofs about critical points}\label{app:critical}

We give a brief proof of Lemma \ref{lmm:critical}

\begin{proof}
Consider the linearized System of \eqref{eq:system2}  at the critical point $(\phi_*,v_*)$  where $\frac{\partial H}{\partial v}(\phi_*,v_*)=\frac{\partial H}{\partial \phi}(\phi_*,v_*)=0$,
\begin{align*}
\begin{cases}
\displaystyle{\frac{d\phi}{d\tau}=\frac{\partial^2 H}{\partial v\partial\phi}(\phi-\phi_*)+\frac{\partial^2 H}{\partial v^2}(v-v_*),}\\
\displaystyle{\frac{dv}{d\tau}=-\frac{\partial^2 H}{\partial v\partial\phi}(v-v_*)-\frac{\partial^2 H}{\partial \phi^2}(\phi-\phi_*)}.
\end{cases}
\end{align*}
The eigenvalues of the linearized system satisfy
$$\lambda^2-\Big(\frac{\partial^2 H}{\partial v\partial\phi}\Big)^2+\frac{\partial^2 H}{\partial\phi^2}\frac{\partial^2 H}{\partial v^2}=0.$$
When 
$$\Big(\frac{\partial^2 H}{\partial v\partial\phi}\Big)^2-\frac{\partial^2 H}{\partial\phi^2}\frac{\partial^2 H}{\partial v^2}<0,$$
 $(\phi_*,v_*)$ is a center point and $H(\phi_*,v_*)$ is a local extremum of $H$.  When 
$$\Big(\frac{\partial^2 H}{\partial v\partial\phi}\Big)^2-\frac{\partial^2 H}{\partial\phi^2}\frac{\partial^2 H}{\partial v^2}>0,$$
$(\phi_*,v_*)$ is a saddle point of System \eqref{eq:system2} and it is a saddle point of $H$.
\end{proof}

We now give a brief proof for Proposition \ref{pro:criticalonhyperbola}

\begin{proof}
Assume that $(\phi_*,v_*)$ is a critical point which satisfies $\phi_*^2-v_*^2=c$. From the second equation in System \eqref{eq:system2}, we obtain  $-2k\phi_*+g=0$. Hence,
$$\phi_*=\frac{g}{2k}~\textrm{ and }~v_*^2=\phi_*^2-c=\frac{g^2-4k^2c}{4k^2} ~\textrm{ when }~g^2-4k^2c>0.$$
Hence,
$$H(\phi_*,v_*)=\frac{c^2}{4}-\frac{c^2}{2}+\frac{g^2}{4k}-\frac{g^2}{2k}=-\frac{c^2}{4}-\frac{g^2}{4k}.$$
$$\frac{\partial^2 H}{\partial v\partial\phi}(\phi_*,v_*)=\frac{g}{k}\sqrt{\frac{g^2-4k^2c}{4k^2}},~~\frac{\partial^2 H}{\partial\phi^2}(\phi_*,v_*)=\frac{g^2+4k^3}{2k^2},~~\frac{\partial^2 H}{\partial v^2}(\phi_*,v_*)=\frac{g^2-4k^2c}{2k^2}.$$
Hence
$$\Big(\frac{\partial^2 H}{\partial v\partial\phi}\Big)^2-\frac{\partial^2 H}{\partial\phi^2}\frac{\partial^2 H}{\partial v^2}=\frac{g^2(g^2-4k^2c)}{4k^4}-\frac{(g^2+4k^3)(g^2-4k^2c)}{4k^4}=-\frac{g^2-4k^2c}{k}$$
Hence, when $k>0$, we have
$$\Big(\frac{\partial^2 H}{\partial v\partial\phi}\Big)^2-\frac{\partial^2 H}{\partial\phi^2}\frac{\partial^2 H}{\partial v^2}<0.$$
By Lemma \ref{lmm:critical}, $(\phi_*,v_*)$ is a center point. When $k<0$, we have
$$\Big(\frac{\partial^2 H}{\partial v\partial\phi}\Big)^2-\frac{\partial^2 H}{\partial\phi^2}\frac{\partial^2 H}{\partial v^2}>0.$$ Lemma \ref{lmm:critical} tells us that $(\phi_*,v_*)$ is a saddle point. 
\end{proof}

\section{Proof of Proposition \ref{pro:structureGam}} \label{app:structure}

\begin{proof}
If $(k, g)=(0, 0)$, all the level sets are of the form $\phi^2-v^2=c$ and the claim is clearly true. Now, we assume $(k, g)\neq (0, 0)$. We use five steps to prove this proposition. In Step 1 and Step 2, we give some useful results. With these arguments, we prove proposition by Step 3, Step 4 and Step 5. In this proof, we always use $\tilde{\alpha}$ to represent the reflection of a curve $\alpha$ about $v=0$.

{\bf Step 1.} 

 $\Gamma_+\cap \Sigma$  ($\Gamma_-\cap\Sigma$)  if nonempty must contain an arc with positive length.

Due to Lemma \ref{lmm:arcs},  $\Gamma_+\cap P_\pm$ are graphs of continuous functions. Because $\Sigma$ is a connected component, when  $\Gamma_+\cap \Sigma\cap P_\pm$  is not empty, it must contain an arc with positive length or it is a single point.

Assume $\Gamma_+\cap \Sigma\cap P_+=\{(\phi_0,v_0)\}$ is a single point. Because $\Sigma$ is a nontrivial connected set, we can see that $\Gamma_+=\{(\phi_0,v_0)\}$.  
This implies that the square root in the definition of $\Gamma_+$ is defined only at $\phi_0$ and at this time we have $\Gamma_-=\Gamma_+$, in which case $\Sigma$ is trivial contradicting with the assumption.  Hence, 
$\Gamma_+\cap \Sigma$ ($\Gamma_-\cap\Sigma$)  if nonempty must contain an arc with positive length. 

{\bf Step 2.}

If there is a point $A=(\phi_*,v_*)\in \Sigma$ such that there are more than three arcs emerging from $A$. The Implicit theorem tells us that $A$ has to be a critical point of $H$ (or System \eqref{eq:system2}).  

If the critical point $A$ is on $\phi^2-v^2=c$,  we know $v_*\neq 0$. Otherwise, we have $\phi_*^2=c$, which is contradict with our assumption.  This point has to be the critical point of the form $\phi_*=\frac{g}{2k}$ for \eqref{eq:system2}. By Proposition \ref{pro:criticalonhyperbola}, this critical point is a center point when $k>0$ and hence we must have $k<0$ (For $k=0$ and $g\neq 0$, this is not a critical point).
In this case, we have 
\[
\sqrt{c^2-4(k\phi_*^2-g\phi_*-h)}=\sqrt{c^2-4k\Big(\phi_*-\frac{g}{2k}\Big)^2+\frac{g^2}{k}+4h}=0.
\] 
 $\Gamma_{\pm}$ are defined both on the right and left of this critical point. As a result,  on either side, there are two curves emerging from $A$. 
 
If $A=(\phi_*, 0)$, since $\phi_*^2\neq c$, $\Gamma_{+}$ and $\Gamma_-$ do not meet here. There are and only are two curves emerging from each side of $A$.

{\bf Step 3.}

If $\Gamma_+\cap \Sigma=\emptyset$ or $\Gamma_-\cap\Sigma=\emptyset$, by Lemma \ref{lmm:arcs}, $\Sigma$ must not intersect $\phi^2-v^2=c$ at $v\neq 0$. By the assumption there is no point $(\phi_0, 0)$ with $\phi_0^2=c$ on it, $\Sigma$ does not intersect the hyperbola on $v=0$ as well.  Moreover, $\Sigma\subset\Gamma_+$ or $\Sigma\subset\Gamma_-.$
As a result, by Lemma \ref{lmm:arcs}, $\Sigma\cap P_+$ if nonempty is the graph of a function $v=f(\phi)$.

If there are  no points on $\Sigma$ as described in Step 2, from the third statement of  Lemma \ref{lmm:arcs} we have three cases for $f(\phi)$: case 1. two endpoints of $f$ are at $v=0$; case 2. one endpoint of $f$ is on $v=0$ the other end tends to infinity; case 3. both ends of $f$ extend to infinity.  Assume $\Sigma_+=\Sigma\cap P_+$, which is given by the graph of function $f$. For case 1 and case 2, by our assumption about $\Sigma$, we know $\Sigma_+\cup\tilde{\Sigma}_+=\Sigma$. At this time $\Sigma$ is a loop for case 1 while both ends tend to infinity for case 2. For case 3, we have $\Sigma=\Sigma_+$.

If there are  points on $\Sigma$ as described in Step 2,  they must be on $v=0$. Otherwise they are on $\phi^2-v^2=0$ and at this time $\Gamma_+\cap\Gamma_-\cap \Sigma\neq\emptyset.$  We can use the lines $\phi=\phi_A$ ($A$ is a critical point in $\Sigma$) to divide the whole plane into finite regions $R_i$.  Consider the arc $\Sigma_i$ of the graph of function $f$ in $R_i$.  The analysis is the same as the last paragraph.

{\bf Step 4}  

Now, we assume $\Gamma_+\cap \Sigma\neq \emptyset$ and $\Gamma_-\cap\Sigma\neq \emptyset$. By Step 1, both parts contain at least an arc with positive length.

Without loss of generality, we can assume  $\Gamma_+\cap \Sigma\cap P_+ \neq \emptyset$. By Lemma \ref{lmm:arcs}, this is a function graph $v=f_1(\phi)$ with the largest domain $D$.  We claim first that $D$ is an interval. Suppose otherwise, then there are two disjoint intervals $I_1$ and $I_2$ such that $v=f_1(\phi)$ has no graphs between them. Let $I_1$ be on the left of $I_2$. Consider the right endpoint of $I_1$, $\phi^*$. If $\Gamma_+$ is defined for $\phi\in (\phi^*, \phi^*+\delta)$, then the graph in $P_+$ must be continuously extended from $f_1(\phi)$ by the expression of $\Gamma_+$ in \eqref{eq:Gammapm}. Hence, $\Gamma_+$ is not defined on $\phi\in (\phi^*, \phi^*+\delta)$ (for $\delta>0$) and consequently,  $\Gamma_-$ is not defined either (the domain of $\Gamma_-$ is contained in the domain of $\Gamma_+$). This means the graphs of $\Gamma_+$ over $I_1$ and $I_2$ cannot be in the same connected component $\Sigma$. 

To be simple, let $\gamma$ be the function graph of $f_1$.  If there are any points $A=(\phi_A,v_A)$ as described in Step 2 on $\gamma$, we  use the lines $\phi=\phi_A$ to divide the whole planes into finite regions $R_i$ (of course, if there are no such points, $R_1$ is the whole plane). Consider the arc $\gamma_i=\gamma\cap R_i$. Assume its end points are $B=(\phi_B,v_B)$ and $C=(\phi_C,v_C)$.   $\gamma_i$ must have one end that ends on $v=0$ or $\phi^2-v^2=c$ with $v\neq 0$ (if both ends tend to infinity, we then have the case in Step 3).  We use the following four cases to study the structure of $\gamma_i.$
\begin{enumerate}
\item  The first case is that one end, $B$, of $\gamma_i$ tends to infinity.
\begin{itemize} 
\item If the other end $C$ satisfies $v_C=0$, then  $C$ is not on $\phi^2-v^2=c$. This means $\Gamma_+$ and $\Gamma_-$ does not meat there. By the symmetry of $\Sigma$, we know $\Sigma\cap R_i=\gamma_i\cup \tilde{\gamma}_i$. In this case, $\Sigma \cap R_i$ has both ends tend to infinity.

\item If the other end $C$ falls onto $\phi^2-v^2=c$, then we must have $v\neq 0$. Then, $\Gamma_+$ and $\Gamma_-$ meet here and $\Gamma_-$ has an arc $\alpha_i$ merging from $C$ which is under $\gamma_i$ in $R_i\cap P_+$. 

If the other end of $\alpha_i$ ends on a point $C'$ with $v_{C'}=0$, then $\gamma_i\cup\alpha_i\cup\tilde{\gamma}_i\cup\tilde{\alpha}_i$ is a simple curve with both ends tending to infinity. 
Further more, if $C'$ happens to be a critical point in Step 2, there are some arc of $\Gamma_-$ beyond this critical points and the discussion for these arcs will be similar as  Step 3. 

If the other end of $\alpha_i$ tends to infinity without intersecting $v=0$, then $\gamma_i\cup\alpha_i=\Sigma\cap R_i\cap P_+$ is a simple curve with both ends tending to infinity. $\Sigma\cap R_i\cap P_-$ if nonempty (this happens only if $C'$ is the critical point on the hyperbola) will be the reflection of this simple curve, which is again a simple curve.
\end{itemize}

\item The second possibility is that both ends are on $v=0$. In this case, by the discussion in Step 2,  $\Gamma_-\cap \Sigma\cap R_i=\emptyset$. As a result, $\Sigma\cap R_i=\gamma_i\cup\tilde{\gamma}_i$ is a loop.

\item Another possibility is one (say $B$) on $v=0$ and one (say $C$) on $\phi^2-v^2=c$ with $v\neq 0$. Then, there is an arc $\alpha_i$ in $\Gamma_-\cap \Sigma\cap R_i$ merging from $C$ in $R_i$ and under $\gamma_i$. The other end of $\alpha_i$ must terminates on a point $C'$ with $v_{C'}=0$. As a result, $\Sigma\cap R_i$ is a loop.  Similarly, if $C'$ happens to be a critical point in Step 2, there are some arc of $\Gamma_-$ beyond this critical points and the discussion for these arcs will be similar as that is Step 3. 

\item Lastly, both $B$ and $C$ are on $\phi^2-v^2=c$ with $v\neq 0$. We call the left endpoint $B$. In this case, $\Gamma_-\cap \Sigma\cap R_i$ exits both on the right of $B$ and on the left of $C$, called $\alpha_i^-$ and $\alpha_i^+$.
\begin{itemize}
\item If $\alpha_i^-=\alpha_i^+$, we call this $\alpha_i$. Then, $\gamma_i\cup \alpha_i=\Sigma\cap R_i\cap P_+$ is a simple curve, which is  a loop. If $\Sigma\cap R_i\cap P_-$ is nonempty (this happens only if $\alpha_i$ touches $v=0$ at a critical point), then, it is the reflection of the loop in $P_+$.

\item If $\alpha_i^-\neq \alpha_i^+$, then, they must both touch $v=0$. As a result, $\gamma_i\cup\alpha_i^-\cup\alpha_i^+=\beta$ and $\beta\cup\tilde{\beta}$ is a simple curve and forms a loop. Further, if these points on $v=0$ happen to be critical points in Step 2, there are some arc of $\Gamma_-$ beyond this critical points and the discussion for these arcs will be similar as that is Step 3. 
\end{itemize}

\end{enumerate}

This then finishes the discussion of the structures of $\Sigma$.

{\bf Step 5}

The last part of the proposition is obvious by Lemma \ref{lmm:arcs} and the expressions of $\Gamma_{\pm}$.
\end{proof}

\section{Peakon weak solutions for the CH equation \eqref{eq:generalCH}}\label{app:CHpeakon}
The definition of weak solution to the CH equation \eqref{eq:generalCH} is given similarly as the mCH equation. $u\in C([0,T);H^1(\mathbb{T}_\ell))$ is a weak solution to \eqref{eq:generalCH} subject to $u(x,0)=u_0(x)$ if the following holds for any $\varphi(x,t)\in C^\infty_c([0,T)\times\mathbb{T}_\ell)$
\begin{multline}
\int_0^T\int_{\mathbb{T}_\ell} u\phi_t dxdt-\int_0^T\int_{\mathbb{T}_\ell} u\phi_{xxt}dxdt+\frac{3}{2}\int_0^T\int_{\mathbb{T}_\ell} u^2\phi_xdxdt+\frac{1}{2}\int_0^T\int_{\mathbb{T}_\ell} u_x^2\phi_xdxdt\\
-\int_0^T\int_{\mathbb{T}_\ell}\frac{1}{2}u^2\phi_{xxx}dxdt=-\int_{\mathbb{T}_\ell} u \phi(x, 0)dx+\int_{\mathbb{T}_\ell} u_0\phi_{xx}(x,0)dx.
\end{multline}
As shown in \cite{Camassa}, the $N$-peakon solutions are of the form
\[
u(x,t)=\sum_{i=1}^Np_i(t)G(x-x_i(t)).
\]
Using the weak formulation, we find that the jump conditions are given by
$x_i(t)$ and $p_i(t)$ satisfy the following ODE system
\begin{gather}
\left\{
\begin{split}
&\frac{dx_i(t)}{dt}=u(x_i(t), t),\\
&\frac{dp_i(t)}{dt}=-p_i(t)\frac{1}{2}[u_x(x_i(t)+)+u_x(x_i(t)-)]. \label{eq:chspeed}
\end{split}
\right.
\end{gather}
Here $p_i(t)$ depends on time and this is different with the $N$-peakon weak solutions to the mCH equation given in Proposition \ref{pro:NpeakonmCH}. 
Notice that for any $x\in\mathbb{R}$, we have
$$-\frac{1}{2}[G_x(x+)+G_x(x-)]=\mathrm{sgn}(x)G(x).$$ 
Hence, \eqref{eq:chspeed} can be changed into
\begin{gather}
\left\{
\begin{split}
&\frac{dx_i(t)}{dt}=\sum_{j=1}^Np_j(t)G(x_i(t)-x_j(t)),\\
&\frac{dp_i(t)}{dt}=\sum_{j=1}^Np_i(t)p_j(t)\mathrm{sgn}(x_i(t)-x_j(t))G(x_i(t)-x_j(t)). \label{eq:chspeed2}
\end{split}
\right.
\end{gather}
This is a Hamiltonian system and the Hamiltonian function is given by
$$\mathcal{H}_0(t)=\frac{1}{2}\sum_{i,j=1}^Np_i(t)p_j(t)G(x_i(t)-x_j(t)).$$
When $\ell=+\infty$, $p_i(t)$ and $x_i(t)$ satisfy the following Hamiltonian system of ODEs:
\begin{align}\label{CH Npeakon}
\begin{cases}
\displaystyle{\frac{d}{dt}x_i(t)=\frac{1}{2}\sum_{j=1}^Np_j(t)e^{-|x_i(t)-x_j(t)|},~~i=1,\ldots,N,}\\
\displaystyle{\frac{d}{dt}p_i(t)=\frac{1}{2}\sum_{j=1}^Np_i(t)p_j(t)\mathrm{sgn}\big(x_i(t)-x_j(t)\big)e^{-|x_i(t)-x_j(t)|},~~i=1,\ldots,N.}
\end{cases}
\end{align}
One can refer to \cite{CamassaLee,holden2006,AlinaLiu} for a rigorous analysis of the Hamiltonian system \eqref{CH Npeakon}.

 If $N=1$, from \eqref{eq:chspeed2} we have
\begin{gather}
p_1(t)\equiv p~\textrm{ and }~\frac{d}{dt}x_1(t)=\frac{p}{2}\coth(\frac{\ell}{2})=c.
\end{gather}
The one peakon weak solution is given by
\begin{align}\label{periodic solitary}
u(x,t)=pG(x-ct).
\end{align}

\bibliographystyle{plain}
\bibliography{bibofperiodicmCH}

\end{document}